\documentclass[12pt]{article}
\usepackage{geometry,float}
\usepackage[parfill]{parskip}
\usepackage{amsmath,amssymb,bbm}
\usepackage{cases}
\usepackage{graphicx,psfrag,epsf}
\usepackage{caption}
\usepackage{subcaption}
\usepackage{tabularx}
\usepackage{microtype}
\usepackage{enumitem}
\usepackage{authblk}
\usepackage{url}
\usepackage{tikz}
\usepackage[colorlinks,citecolor=blue,urlcolor=blue,linkcolor=blue,linktocpage=true]{hyperref}
\usepackage[amsmath,amsthm,thmmarks]{ntheorem}
\usepackage{algpseudocode,algorithm}

\usepackage{cleveref}
\crefformat{equation}{(#2#1#3)}
\crefrangeformat{equation}{(#3#1#4) to~(#5#2#6)}
\crefname{equation}{}{}
\Crefname{equation}{}{}
\usepackage[authoryear]{natbib}

\newtheorem{theorem}{Theorem}

\newtheorem{example}{Example}
\newtheorem{corollary}[theorem]{Corollary}
\newtheorem{remark}{Remark}
\newtheorem{definition}{Definition}
\crefname{definition}{\textbf{definition}}{definitions}
\Crefname{definition}{Definition}{Definitions}
\crefname{assumption}{\textbf{assumption}}{assumptions}
\Crefname{assumption}{Assumption}{Assumptions}

\let\hat\widehat
\let\tilde\widetilde

\newcommand{\R}{\mathbb R}
\newcommand{\E}{\mathbb E}

\newcommand{\D}{\mathbb D}

\newcommand{\mcX}{\mathcal X}

\newcommand{\mcR}{\mathcal R}
\newcommand{\mcF}{\mathcal F}
\newcommand{\mcG}{\mathcal G}
\newcommand{\mcB}{\mathcal B}

\newcommand{\cc}{\text{\textsc{cc}}}
\newcommand{\ac}{\text{\textsc{ac}}}
\newcommand{\nc}{\text{\textsc{nc}}}
\newcommand{\assum}{A}

\renewenvironment{proof}{{\bf Proof.}}{$\Box$}

\usepackage{color}

\newcommand{\blind}{1}

\begin{document}

\def\spacingset#1{\renewcommand{\baselinestretch}%
{#1}\small\normalsize} \spacingset{1}


\if1\blind { \title{\bf Nonparametric Pattern-Mixture Models for Inference with Missing Data} 
\author{Yen-Chi Chen\thanks{yenchic@uw.edu} and Mauricio Sadinle\thanks{msadinle@uw.edu}\\
University of Washington}
		\date{}
  \maketitle
} \fi

\if0\blind
{
  \title{\bf }
  \author{}
	\date{}
  \maketitle
} \fi

\begin{abstract}
Pattern-mixture models provide a transparent approach for handling missing data, where the full-data distribution is factorized in a way that explicitly shows the parts that can be estimated from observed data alone, and the parts that require identifying restrictions.  We introduce a nonparametric estimator of the full-data distribution based on the pattern-mixture model factorization.  Our approach uses the empirical observed-data distribution and augments it with a nonparametric estimator of the missing-data distributions under a given identifying restriction.  Our results apply to a large class of \emph{donor-based} identifying restrictions that encompasses commonly used ones and can handle both monotone and nonmonotone missingness.  We propose a Monte Carlo procedure to derive point estimates of functionals of interest, and the bootstrap to construct confidence intervals.
\end{abstract}

\noindent%
  \vfill

\newpage
\spacingset{1.45} 

\section{Introduction}

Statistical inference with missing data requires handling the joint distribution of the study variables and their missingness indicators \citep{Rubin76}.  Two natural ways of factorizing this distribution arise.  The most popular approach is the selection model factorization, where we work with the product of the marginal distribution of the study variables and the missingness mechanism.  The popularity of this approach is at least partly due to the fact that analysts are able to use the same, typically parametric, models that they would use with complete data, and that they do not have to handle the missingness mechanism under the assumption of ignorability \citep[e.g.,][Chapter 6.2]{LittleRubin02}.  Nevertheless, selection models are difficult to work with under nonignorable missing data, as ensuring that they are identifiable is difficult in general \citep[e.g.,][]{sadinle2019SAN} and it often heavily relies on parametric assumptions \citep[e.g.,][p. 107]{DanielsHogan08}.

Pattern-mixture models offer an alternative, where the full-data distribution is factorized as the  distribution of the observed data times the conditional distribution of the missing data given the observed data \citep{Little93}.  This factorization shows that the full-data distribution is not identifiable, and it clearly separates what  can and cannot be recovered from observed data alone.  Specifically, the observed-data distribution can be directly estimated  from the observed data, but the missing-data distribution is only obtainable after imposing identifying restrictions that tie the full-data distribution to the observed-data distribution.  Despite this transparent separation of what one can and cannot recover from data alone, implementations of pattern-mixture models have predominantly been done under parametric assumptions for the observed-data distribution \citep[e.g.,][]{Little93,Little94,Little95,HoganLaird97,DanielsHogan00,Fitzmaurice03,Kenwardetal03,RoyDaniels08,WangDaniels11}, although more flexible Bayesian nonparametric modeling approaches have recently been developed \citep{LineroDaniels15,Linero17,LineroDaniels18}.  Existing approaches also mainly focus on handling monotone missingness in longitudinal studies.  

In this article we propose a methodology for nonparametric inference under pattern-mixture models.  We show how inferences can be derived under a very general class of identifying restrictions that can handle both monotone and nonmonotone missingness.  Our proposed procedure guarantees that inferences only depend on the observed data and a given identifying restriction. 





{\emph{Main Contributions.}}
\begin{enumerate}
\item 
We introduce the concept of \emph{donor-based} identification in Section~\ref{sec:PMMs}, as a way of unifying different identification restrictions commonly used to handle monotone missingness.  The application of pattern-mixture models to nonmonotone nonresponse has been prevented partly due to the lack of identifying restrictions applicable to that more general scenario.  Here we also extend donor-based identification to handle nonmonotone missingness.

\item 
We propose nonparametric estimators of the full-data distribution, where the part corresponding to the observed-data distribution is estimated using its empirical version, and the missing-data distributions are estimated under identifying restrictions using a \emph{surrogate estimate} of the observed-data distribution that relies on conditional kernel-density estimators (Section~\ref{sec::est::full}).  


\item
We show how to estimate statistical functionals of interest (Section~\ref{sec::functional}), where the point estimator is obtained  numerically using a Monte Carlo approach (Section~\ref{sec::MCE}), which we prove to be the valid (Theorem~\ref{thm::MC_icin}). 

\item We introduce a bootstrap approach for constructing confidence intervals
for statistical functionals of interest (Section~\ref{sec::bootstrap}) and prove its validity (Theorem~\ref{thm::LBT}). 

\item Finally, we use \emph{Efron's bootstrap diagram} to make explicit the estimand that our method is inferring (Section~\ref{sec::BD}). 
This diagram helps clarify the validity of our procedures even under misspecified identifying assumptions. 


\end{enumerate}



\emph{Outline.}
We give a general introduction to pattern-mixture models, and present some generalizations of commonly used identifying restrictions in Section~\ref{sec:PMMs}.
We introduce our nonparametric estimator 
along with a Monte Carlo approach for parameter estimation in Section~\ref{sec::estimate}. 
In Section~\ref{sec::bootstrap},
we explain how one can use the bootstrap approach to construct valid confidence intervals.
Theoretical results, including convergence rates, asymptotic normality,
validity of Monte Carlo methods,
and the validity of bootstrap methods,
are given in Section~\ref{sec::theory}. 
We provide a data analysis in Section~\ref{sec::data}. 
Finally, we give a short discussion about future directions in Section~\ref{sec::discussion}.

\section{Pattern-Mixture Models}	\label{sec:PMMs}

\subsection{Setup}

We consider $d$ study variables $X=(X_1,\ldots,X_d)$, taking values on a sample space $\mcX=\otimes_{j=1}^d\mcX_j$, and its vector of response indicators $R=(R_1,\ldots,R_d)$, taking values on $\mcR\subseteq \{0,1\}^d$, with $R_j=0$ when variable $j$ is missing and $R_j=1$ when it is observed.  An element $r=(r_1,\ldots,r_d)\in \{0,1\}^d$ is called a response pattern, which we will often represent as the string $r_1\ldots r_d$.  Given $r\in\mcR$, we define $\bar{r}=1_d-r$ to be the missingness pattern, where $1_d = (1,1,\cdots, 1)$ is a vector of ones of length $d$.  For a response pattern $r$, we define $X_{\bar r}=(X_j: r_j=0)$ to be the missing variables and $X_{r}=(X_j: r_j=1)$ to be the observed variables, which have sample spaces $\mcX_{\bar r}$ and $\mcX_{r}$, respectively.  For example, if $X=(X_1,X_2,X_3)$ and $r=101$, that is, the realizations of random variables $X_1$ and $X_3$ are observed, then $X_{101}=(X_1,X_3)$ and $X_{\overline{101}}=X_{010}=X_2$.  

We refer to the true joint distribution $F$ of $X$ and $R$ as the full-data distribution, and denote the collection of full-data distributions by $\mcF$.  We denote the distribution of $X$ given $R=r$ by $F_r$.  We write $f(x\mid r)$ for the density of $F_r$ with respect to an appropriate dominating measure $\mu$, and  
 $f(x,r)$ for the density of $F$ with respect to the product of $\mu$ and the counting measure on $\{0,1\}^d$.  $F$ cannot be estimated from observed data alone, since when $R=r$ we only get to see the realization of $X_r$.  The observed-data distribution $G$ involves the response indicators and the corresponding  observed study variables, with density derived as $g(x_r,r)\equiv \int_{\mcX_{\bar r}} f(x,r) \mu(dx_{\bar r})$.  
We denote the collection of observed-data distributions by $\mcG$.  
Throughout the document we more generally use $g$ to denote density functions that can be directly derived from $G$, and $f$ for density functions that more generally depend on $F$ and not exclusively on $G$.  
Namely, quantities denoted using $g$ can be derived
from the observed-data distribution alone, whereas quantities denoted using $f$ will depend on  identifying assumptions. 
For example, we write 
\begin{equation}\label{eq:fgf}
f(x,r)=g(r)f(x\mid r)=g(x_r,r)f(x_{\bar r}\mid x_r,r),
\end{equation}
where $g(r)=\int_{\mcX_{r}} g(x_r,r) \mu(dx_{r})$, and $f(x_{\bar r}\mid x_r,r)$ is the density of the distribution of $X_{\bar r}\mid X_r=x_r,R=r$, which represents the non-identifiable parts of the full-data distribution, referred to as the extrapolation or missing-data distribution \citep[][p. 166]{DanielsHogan08}.

%


\cite{Little93} refers to models based on the factorization given by \eqref{eq:fgf} as pattern-mixture models, since the marginal distribution of $X$ is obtained as a mixture of the distributions of $X$ given each response pattern, that is, $$f(x)=\sum_r g(r)f(x\mid r)=\sum_r g(x_r, r)f(x_{\bar r}\mid x_r,r).$$ 

\subsection{Identifying assumptions and modeling}

The pattern-mixture factorization explicitly shows the parts of the full-data distribution that cannot be identified from data, represented by the extrapolation densities $f(x_{\bar r}\mid x_r,r)$.  The role of an identifying restriction or assumption $A$ is to map an observed-data distribution $G$ into  a full-data distribution $F_A \equiv A(G)\in \mcF$.  Under the pattern-mixture approach, based on \eqref{eq:fgf}, identifying assumptions amount to indicating how to construct extrapolation distributions based on the observed-data distribution.  In terms of densities, we write $f_A\equiv A(g)$ to indicate that we obtain $f_A(x_{\bar r}\mid x_r,r)$, and therefore $f_A(x,r)$ in \eqref{eq:fgf}, under restriction $A$.


Parametric implementations of pattern-mixture models work under a restricted family of observed-data distributions $\mathcal{M}\subset \mcG$; for example, \cite{Little93} used multivariate Gaussian distributions for each $F_r$.  The combination of modeling of the observed data and specification of an identifying restriction leads to three types of parameters that we need to distinguish, here depicted in Figure \ref{fig:parameters}.  The observed-data distribution $G$ is obtained from a true full-data distribution $F$ under a true missingness mechanism that implicitly defines a mapping $A_0$ such that $F=A_0(G)$, and correspondingly, $\theta_0\equiv \theta(F)\equiv \theta(A_0(G))$ is the true value of a population parameter of interest.  As we have argued, $A_0$ cannot be learned from $G$ alone, and therefore we must work under an assumption $A$ on the missing-data mechanism that conceptually maps $G$ into $A(G)$ in the space of full-data distributions.  The population parameter of interest under assumption $A$ is then $\theta\!_{A}\equiv \theta(A(G))$.  Furthermore, if one uses a parametric approach restricting the observed-data distribution to a class $\mathcal{M}\subset \mcG$, and if $G$ is not contained in $\mathcal{M}$, all we can hope to recover is the distribution $G\!_{\mathcal{M}}$ in $\mathcal{M}$ that diverges the least from $G$, for example in a Kullback-Leibler sense \citep[e.g.,][p. 55]{van1998asymptotic}.  From that distribution, we could then obtain a full-data distribution $A(G\!_{\mathcal{M}})$ under a missing-data restriction $A$, thereby leading to the population parameter $\bar\theta\!_{A}\equiv \theta(A(G\!_{\mathcal{M}}))$.  If the model $\mathcal{M}$ contains the true $G$, then the right and center branches of Figure \ref{fig:parameters} are the same.  If the identifying assumption $A$ is correct, then the left and center branches of Figure \ref{fig:parameters} are the same.  If  model $\mathcal{M}$ and identifying restriction $A$ are both correct then the three branches coincide.  However, we can never recover $A_0$ from observed data alone, so the gap between the left and middle branches of Figure \ref{fig:parameters} completely depends on assumption $A$ and cannot be reduced based on observed-data.  
It is however possible to asymptotically derive inferences under $G$, by modeling it in a nonparametric way.  Our goal can then be seeing as closing the gap between the middle and right branches of Figure \ref{fig:parameters}. 

We now present a unified view and a generalization of commonly used identifying assumptions for pattern-mixture models. 

\begin{figure}
\center
\includegraphics[width=.6\linewidth]{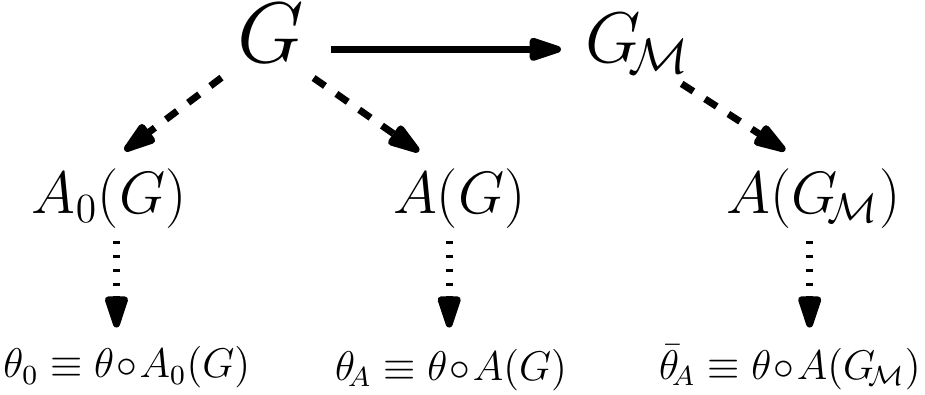}
\caption{Population parameters of interest.  Solid arrow: modeling; dashed arrow: mapping implied by missing-data restriction; dotted arrow: derivation of parameter of interest.}
\label{fig:parameters}
\end{figure}


\subsection{Identifying assumptions under monotone missingness}\label{ss:ident_monotone}

In the context of a longitudinal study, where the study variables $X_1,\dots,X_d$ correspond to measurements over time, it is common to obtain monotone missingness, that is, when missingness in $X_t$ implies missingness in every $X_{t'}$ for $t'>t$, due to dropout of study participants \citep[e.g.,][p. 88]{DanielsHogan08}.  
In such situations, the response pattern $R$ is uniquely determined by the dropout time $T$, that is, $T\equiv |R|$, the number of non-zero entries of $R$.  Given a response pattern $r$ such that $|r|=t$, we denote 
$x_r = (x_1,\cdots,x_{t}) \equiv x_{\leq t}$; likewise $x_{<t} \equiv (x_{1},\cdots, x_{t-1})$ and $x_{>t} \equiv (x_{t+1},\cdots, x_{d})$. Then, the extrapolation density can be written as
\begin{equation}\label{eq:product}
f(x_{\bar r}\mid x_r, r) = f(x_{>t}\mid x_{\leq t}, T = t) = \prod_{s=t+1}^{d} f(x_{s}\mid x_{<s}, T= t).
\end{equation}
Therefore, under monotone missingness, the specification of identifying restrictions for pattern-mixture models amounts to indicating a way of deriving $f(x_{s}\mid x_{<s}, T= t)$ as a function of the observed-data distribution, for each $t=0,\dots,d-1$ and $s=t+1,\dots,d$. 

We now present a unifying framework that encompasses common approaches to identify $f(x_{s}\mid x_{<s}, T= t)$, for $s>t$, which assumes that this conditional distribution is the same as the analogous distribution among observations with later dropout times, called \emph{donors}.  
\begin{definition}[Donor-based identification for monotone missingness]\label{def:donor-restr}
For each $t$ and $s$ with $t<s$, let $\assum_{ts}\subseteq\{s,s+1,\cdots, d\}$ denote a set of dropout times called \emph{donor times}.  A donor-based identification restriction $A$ sets  
\begin{equation}\label{eq:monotone_donor}
f_\assum(x_{s}\mid x_{<s}, T=t) \equiv g(x_{s}\mid x_{<s}, T \in \assum_{ts}).
\end{equation}
\end{definition}
The observations with dropout times in $\assum_{ts}$ can be thought of as distribution-donors, because they are such that their values of $X_1,\dots,X_s$ are observed, and therefore we can use them to obtain the right-hand side of \eqref{eq:monotone_donor}.   The notation $\assum_{ts}$ indicates that this set is associated with an identifying restriction $A$ which can use different sets of donor times for each $t$ and $s$.  
\begin{remark}[Donor-based selection-model restriction]\label{r:selection}
For selection models, Definition \ref{def:donor-restr} can be equivalently formulated in terms of restrictions of the missingness mechanism, namely, assuming \eqref{eq:monotone_donor} is equivalent to assuming that 
$f(T=t\mid x_{\leq s})/f(T\in \assum_{ts} \mid x_{\leq s})$ 
is constant as a function of $x_{s}$.  
\end{remark}
Donor-based restrictions encompass different popular identifying restrictions, which can be specified via different subsets $\assum_{ts}$.
\begin{example}[Complete-case] \cite{Little93} proposed the complete-case (CC) restriction, under which identification is tied to the completely observed cases only, that is, $\cc_{ts}= \{d\}$ for all $s> t$, leading to
$
f_{\cc}(x_{s}\mid x_{<s}, T=t) \equiv g(x_{s}\mid x_{<s}, T=d),
$ or, in terms of the missingness mechanism, it assumes that $f(T=t\mid x_{\leq s})/f(T=d \mid x_{\leq s})$ is constant as a function of $x_s$, for all $s$ and $t$ with $s> t$.
\end{example}
\begin{example}[Available-case] 
\cite{ACMV} proposed the available-case (AC) restriction, where we use all available cases to identify the missing-data distribution, that is, $\ac_{ts} = \{s,s+1,\cdots,d\}$ for all $s> t$, and 
$f_{\ac}(x_{s}\mid x_{<s}, T=t) \equiv g(x_{s}\mid x_{<s}, T\geq s)$.  Equivalently, this assumes that $f(T=t\mid x_{\leq s})/f(T\geq s \mid x_{\leq s})$ is constant as a function of $x_s$, for all $s$ and $t$ with $s> t$.  Under monotone missingness, this assumption is equivalent to the missing at random assumption of \cite{Rubin76}.
\end{example}
\begin{example}[Neighboring-case]\label{ex:NC}
\cite{Thijs} introduced the neighboring-case (NC) restriction, which uses the nearest case available for identification, namely, $\nc_{ts} = \{s\}$ for all $s> t$, and  
$f_{\nc}(x_{s}\mid x_{<s}, T=t) \equiv g(x_{s}\mid x_{<s}, T=s)$.  This corresponds to assuming that $f(T=t\mid x_{\leq s})/f(T=s \mid x_{\leq s})$ does not change with $x_s$,  for all $s> t$.  Under monotone missingness, this is equivalent to the itemwise conditionally independent nonresponse assumption of \cite{sadinle2017itemwise}.
\end{example}
\begin{example}[$k$-nearest-case]
Other reasonable restrictions are possible under our general formulation in \eqref{eq:monotone_donor}.  For example, one could think of a $k$-nearest-case ($k$NC) subclass of restrictions, where $k\nc_{ts} = \{s,\cdots,\min\hspace{-2pt}\text{\small\{}s+k, d\text{\small\}}\}$, and
$f_{k\nc}(x_{s}\mid x_{<s}, T=t) \equiv g(x_{s}\mid x_{<s}, s\leq T\leq s+k)$.  This is equivalent to assuming that $f(T=t\mid x_{\leq s})/f(s\leq T\leq s+k \mid x_{\leq s})$ does not depend on $x_s$,  for all $s$ and $t$ with $s> t$.
\end{example}
Note that the sets $\assum_{ts}$ in the particular cases presented above are the same for all $t$, although they change with $s$.  However, in the formulation given by Definition \ref{def:donor-restr} these sets  could change over both $t$ and $s$.  

\begin{remark}[More general restrictions]\label{r:moregeneral} We also point out that it is possible to devise even more general strategies for identification compared to our donor-based approach; for example, for each $t$ and $s$, $t<s$, we could choose different sets $\assum_{ts}^{(u)}$, and then identify $f_\assum(x_{s}\mid x_{<s}, T=t)=\sum_{u}\omega_{ts}^{(u)} g(x_{s}\mid x_{<s}, T \in \assum_{ts}^{(u)})$, for some set of weights $\{\omega_{ts}^{(u)}\}_u$ that adds up to 1.  This is similar to a general strategy presented by \cite{Thijs}, but such approaches in general lack the interpretability in terms of the missingness mechanism given by Remark \ref{r:selection}.
\end{remark}

\subsection{Identifying assumptions under nonmonotone missingness}\label{ss:ident_nonmonotone}

Most existing identifying restrictions for pattern-mixture models have been developed to handle monotone missingness.  Among the restrictions mentioned in the previous section, the CC restriction is the only one that is readily applicable for nonmonotone nonresponse.  
Under that restriction, the distribution of the missing variables $X_{\bar r}$ given the observed data $(X_r, R=r)$ is the same as the distribution of $X_{\bar r}$ given $(X_r, R=1_d)$, coming from fully-observed responses, regardless of the response pattern $r$.  We obtain
$f_\cc(x_{\bar r}\mid x_{r}, r) \equiv g(x_{\bar r}\mid x_{r}, 1_d)$, 
for all $r\in \{0,1\}^d$.  Other donor-based restrictions, which are intuitively reasonable under monotone missingness, do not currently have a clear analog in the nonmonotone case.  

We note that donor-based restrictions for monotone missingness build on the factorization of the extrapolation density in \eqref{eq:product}, which naturally follows the longitudinal order of the study variables, which in turn defines the monotonicity in the response patterns.  
With nonmonotone missingness, 
to identify a full-data distribution, we need to identify $f(x_{\bar r}\mid x_r, r)$ for each response pattern $r\in\{0,1\}^d$. Each $f(x_{\bar r}\mid x_r, r)$ can be factorized as a product of sequential conditional densities analogously to \eqref{eq:product}, but the order in which the variables appear in the factorization do not have to follow the indexing of the study variables, especially if that indexing is arbitrary.  For example, with $d=4$ variables and response pattern $r=0101$, we need to identify $f(x_{\bar r}\mid x_r, r)=f(x_1,x_3\mid x_2,x_4,R=0101)$, which can be factorized as $f(x_3\mid x_2,x_4,R=0101)f(x_1\mid x_2,x_3,x_4,R=0101)$ or as $f(x_1\mid x_2,x_4,R=0101)f(x_3\mid x_1,x_2,x_4,R=0101)$.  
 
Denote the index set of the study variables by $[d]= (1,\dots, d)$, and let $[d]_{r}=(j: r_j=1)$ and $[d]_{\bar r}=(j: r_j=0)$ be the indices of the observed and missing variables according to $r$, respectively. 
Consider a permutation $\pi^{(r)}$ of $[d]$ such that its first $|r|$ entries, $\pi^{(r)}_{\leq |r|}$, contain the indices of the entries of $r$ that equal 1, $[d]_{r}$.  For example, with $r=0101$, $\pi^{(r)}$ is a permutation of $[4]$ such that its first $|r|=2$ entries equal $2$ and $4$, that is, either $\pi^{(r)}=(2,4,1,3)$ or $\pi^{(r)}=(2,4,3,1)$. Denote 
the $j$th entry of $\pi^{(r)}$ by $\pi^{(r)}_j$. 
The variables $X$ can be reordered using the permutation $\pi^{(r)}$ as $X_{\pi^{(r)}}=(X_{\pi^{(r)}_j}:j\in [d])$, and 
we can use this to base the factorization 
\begin{equation}\label{eq:fact_perm} 
f(x_{\bar r}\mid x_r, r)
=\prod_{j=|r|+1}^{d} f(x_{\pi^{(r)}_j}\mid x_{\pi^{(r)}_{<j}}, r),
\end{equation}
where $x_{\pi^{(r)}_{<j}}$ denotes a value of the $X$ variables with indices corresponding to the first $j-1$ entries of $\pi^{(r)}$. For example, with $r=0101$ and $\pi^{(r)}=(2,4,3,1)$, $X_{\pi^{(r)}}=(X_2,X_4,X_3,X_1)$, and \eqref{eq:fact_perm} corresponds to $f(x_1,x_3\mid x_2,x_4,R=0101)=f(x_3\mid x_2,x_4,R=0101)f(x_1\mid x_2,x_3,x_4,R=0101)$.  Note that with the notation used above, the observed variables according to $r$ are $X_r=X_{\pi^{(r)}_{\leq |r|}}$.


Based on the factorization in \eqref{eq:fact_perm}, we need to identify 
$f(x_{\pi^{(r)}_j}\mid x_{\pi^{(r)}_{<j}}, r)$ for each $j>|r|$. 
To do this, we need to use the response patterns where all of the variables in $X_{\pi^{(r)}_{\leq j}}$ are observed. 
Denote $1[\pi^{(r)}_{\leq j}]$ the response pattern with ones in the entries $\pi^{(r)}_{\leq j}$.  We write $r\preceq r'$ if $r'$ indicates at least the same observed variables as $r$.  
Denote $\mathcal{D}(\pi^{(r)}_{\leq j})\equiv \{r': 1[\pi^{(r)}_{\leq j}]\preceq r'\}$ as the set of response patterns where all of the variables in $X_{\pi^{(r)}_{\leq j}}$ are observed, and therefore $\mathcal{D}(\pi^{(r)}_{\leq j})$ represents the set of potential donors for identification of $f(x_{\pi^{(r)}_j}\mid x_{\pi^{(r)}_{<j}}, r)$.  
For example, with $r=0100$ and $\pi^{(r)}=(2,4,3,1)$, we have $\pi^{(r)}_{\leq 2}=(2,4)$, $\pi^{(r)}_{\leq 3}=(2,4,3)$ and $\pi^{(r)}_{\leq 4}=(2,4,3,1)$; $1[\pi^{(r)}_{\leq 2}]=0101$, $1[\pi^{(r)}_{\leq 3}]=0111$ and $1[\pi^{(r)}_{\leq 4}]=1111$; and finally $\mathcal{D}(\pi^{(r)}_{\leq 2})=\{0101,1101,0111,1111\}$, $\mathcal{D}(\pi^{(r)}_{\leq 3})=\{0111,1111\}$ and $\mathcal{D}(\pi^{(r)}_{\leq 4})=\{1111\}$.  Donor-based identification strategies use subsets of the possible sets of donors that can be used for identification.  




\begin{definition}[Donor-based identification for nonmonotone missingness]\label{def:path-donor}
Given a response pattern $r\in\{0,1\}^d$, let $\pi^{(r)}$ be a permutation of $[d]$ such that its first $|r|$ entries contain the indices $[d]_{r}$, that is, $\pi^{(r)}_{\leq |r|}=[d]_{r}$. 
Let $\mathcal{D}(\pi^{(r)}_{\leq j})\equiv \{r': 1[\pi^{(r)}_{\leq j}]\preceq r'\}$ be the set of response patterns where the $X$ variables with indices $\pi^{(r)}_{\leq j}$ are observed.  
Let $\assum(r,\pi^{(r)}_{\leq j})\subseteq \mathcal{D}(\pi^{(r)}_{\leq j})$ denote a set of response patterns called \emph{donor patterns}, for $j=|r|+1,\dots,d$.  A \emph{donor-based identification restriction} $A$ sets  
\begin{equation}\label{eq:path_donor}
f_\assum(x_{\pi^{(r)}_j}\mid x_{\pi^{(r)}_{<j}}, R=r) \equiv g(x_{\pi^{(r)}_j}\mid x_{\pi^{(r)}_{<j}}, R\in\assum(r,\pi^{(r)}_{\leq j})).
\end{equation}
\end{definition}

As in the monotone case, the observations with response patterns in $\assum(r,\pi^{(r)}_{\leq j})$ can be thought of as donor cases, because they are such that their values of $X_{\pi^{(r)}_{\leq j}}$ are observed, and therefore we can use them to obtain the right-hand side of \eqref{eq:path_donor}.  
Likewise, as in the monotone case, donor-based restrictions can also be redefined in terms of selection models. 
\begin{remark}[Nonmonotone donor-based selection-model restriction]\label{r:pathdonor_selection}
Definition \ref{def:path-donor} can be equivalently formulated for selection models in terms of restrictions of the missingness mechanism, namely, \eqref{eq:path_donor} is equivalent to assuming that 
$f(R=r\mid x_{\pi^{(r)}_{\leq j}})/f(R \in \assum(r,\pi^{(r)}_{\leq j}) \mid x_{\pi^{(r)}_{\leq j}})$ 
is constant as a function of $x_{\pi^{(r)}_{j}}$.  
\end{remark}

We now explain how the CC restriction fits into this general formulation and present some examples that generalize the AC, NC and $k$NC restrictions to nonmonotone nonresponse.
  
\begin{example}[Complete-case] Regardless of the values of $r$, $\pi^{(r)}$, and $j$, the set of possible donors $\mathcal{D}(\pi^{(r)}_{\leq j})\equiv \{r': 1[\pi^{(r)}_{\leq j}]\preceq r'\}$ always contains the pattern $1_d$.  Taking all sets of donor patterns to contain $1_d$ only, that is, $\cc(r,\pi^{(r)}_{\leq j})=\{1_d\}$, leads to the CC restriction, $f_\cc(x_{\bar r}\mid x_{r}, r) \equiv g(x_{\bar r}\mid x_{r}, 1_d)$.  Note that this holds regardless of the pattern $r$ and is invariant to the permutation $\pi^{(r)}$.
\end{example}

\begin{example}[Available-case] We can also use the full sets of available donor patterns for identification, that is, $\ac(r,\pi^{(r)}_{\leq j})= \mathcal{D}(\pi^{(r)}_{\leq j})$, leading to 
\[f_\ac(x_{\pi^{(r)}_j}\mid x_{\pi^{(r)}_{<j}}, R=r) \equiv g(x_{\pi^{(r)}_j}\mid x_{\pi^{(r)}_{<j}}, R\succeq 1[\pi^{(r)}_{\leq j}]),\]
for all $r$, $\pi^{(r)}$ and $j>|r|$.  For example, with $r=0100$ and $\pi^{(r)}=(2,4,3,1)$, we obtain 
$f_\ac(x_{4}\mid x_{2}, R=0100) \equiv g(x_{4}\mid x_{2}, R\succeq 0101)$, $f_\ac(x_{3}\mid x_{2},x_4, R=0100) \equiv g(x_{3}\mid x_{2},x_4, R\succeq 0111)$, and $f_\ac(x_{1}\mid x_{2},x_3,x_4, R=0100) \equiv g(x_{1}\mid x_{2},x_3,x_4, R=1111)$.  We notice that, naturally, different factorizations based on different permutations $\pi^{(r)}$ will lead to different identified full-data distributions.
\end{example}

\begin{example}[Neighboring-case]\label{ex:NC}  The pattern $1[\pi^{(r)}_{\leq j}]$ is the closest pattern to $r$ under which all the variables $X_{\pi^{(r)}_{\leq j}}$ are available, and therefore it represents observations with the minimal variables needed to identify $f(x_{\pi^{(r)}_j}\mid x_{\pi^{(r)}_{<j}}, R=r)$.  Therefore, it is reasonable to take $\nc(r,\pi^{(r)}_{\leq j})=\{1[\pi^{(r)}_{\leq j}]\}$ for all $r$, $\pi^{(r)}$ and $j>|r|$.  For example, with $r=0100$ and $\pi^{(r)}=(2,4,3,1)$, we obtain 
$f_\nc(x_{4}\mid x_{2}, R=0100) \equiv g(x_{4}\mid x_{2}, R=0101)$, $f_\nc(x_{3}\mid x_{2},x_4, R=0100) \equiv g(x_{3}\mid x_{2},x_4, R= 0111)$, and $f_\nc(x_{1}\mid x_{2},x_3,x_4, R=0100) \equiv g(x_{1}\mid x_{2},x_3,x_4, R=1111)$.  Here again, the final identified full-data distribution depends on the chosen factorization based on permutation $\pi^{(r)}$.
\end{example}
\begin{example}[$k$-nearest-case] We can define the donor patterns within $k$ distance of $1[\pi^{(r)}_{\leq j}]$ as $\mathcal{D}_k(\pi^{(r)}_{\leq j})\equiv \{r': 1[\pi^{(r)}_{\leq j}]\preceq r', H(1[\pi^{(r)}_{\leq j}], r')\leq k\}$, where $H(r,r')$ is the Hamming distance between two response patterns, which counts the number of entries where $r$ and $r'$ disagree.  Then, the $k$NC restriction takes $k\nc(r,\pi^{(r)}_{\leq j})= \mathcal{D}_k(\pi^{(r)}_{\leq j})$.  
\end{example}

\begin{remark}[More general restrictions]\label{r:moregeneral2} Along the lines of Remark \ref{r:moregeneral}, here we also point out that it is possible to devise even more general identification strategies compared to our donor-based approach.  In Definition \ref{def:path-donor} we tie identification to a set of donors, but a more general approach would be to take a mixture of 
$g(x_{\pi^{(r)}_j}\mid x_{\pi^{(r)}_{<j}}, R\in\assum^{(u)}(r,\pi^{(r)}_{\leq j}))$
across $u$, as in Remark \ref{r:moregeneral}, for fixed permutations $\{\pi^{(r)}\}_r$.  Yet, an even more general approach is to identify multiple versions of $f(x_{\bar r}\mid x_r, r)$ based on different factorizations induced by multiple sets of permutations $\{\pi^{(r)}\}_r$, and then take a weighted average of them.  Again, these approaches lack the interpretability in terms of the missingness mechanism given by Remark \ref{r:pathdonor_selection}.
\end{remark}

\subsection{Nonparametric identification and sensitivity analysis}

The previous sections indicate that, with pattern-mixture models, we start with an observed-data density $g(x_r,r)$, and complement it with extrapolation densities $f_A(x_{\bar r}\mid x_r,r)$ obtained from an identifying restriction $A$, leading to a full-data distribution with density $f_A(x,r)$.  This construction is such that $\int_{\mcX_{\bar r}} f_A(x,r) \mu(dx_{\bar r})= g(x_r,r)$, that is, the observed-data distribution implied by the full-data distribution under $A$ is the same as the true observed-data distribution, meaning that the identifying restrictions only impose constraints on what cannot be recovered from observed data.  This implies that the full-data distribution under $A$ is observationally equivalent to the true full-data distribution, and it means that assumption $A$ cannot be tested based on the observed data \citep[e.g.,][]{sadinle2019SAN}.  This desirable characteristic is known as nonparametric identification, nonparametric saturation or just-identification \citep{Robins97,Vansteelandtetal06, DanielsHogan08, HoonhoutRidder18}. 

The choice of the appropriate identifying restriction is no trivial matter, as it cannot be justified based on the observed data.  Regardless of whether one feels confident about an identifying restriction, it is desirable to perform sensitivity analyses, where inferences are obtained under different identifying restrictions \citep{Scharfstein18}.  The advantage of the construction that we have presented is that the same observed-data distribution can be used to obtain different full-data distributions under different identifying restrictions.  Since all of such full-data distributions will be observationally equivalent, discrepancies in inferences will be entirely due to the different identifying restrictions, making pattern-mixture models ideal for conducting sensitivity analyses.

\section{Nonparametric Estimation}\label{sec::estimate}

As mentioned before, the observed-data distribution $G$, with density $g(x_r, r)$, is directly estimable from the observed data. 
The observed data $S_n=(X_{i,R_i},R_i)_{i=1}^n$ corresponds to a random sample from $G$.  The extrapolation densities, represented by $f(x_{\bar r}\mid x_r,r)$, have to be recovered under an identifying assumption $A$, which provides a recipe for writing $f_A$ in terms of $g$, and we write $f_A=A(g)$.  
Definitions \ref{def:donor-restr} and \ref{def:path-donor} provide explicit forms of $f_A = A(g)$ under donor-based identification. 
Under $A$, the full-data density is given by
$$
f_A(x,r)=g(x_r, r)f_A(x_{\bar r}\mid x_r,r),
$$
and
 the associated full-data distribution is given by 
\begin{equation}\label{eq:FA}
\begin{aligned}
F_A(B,r) 
& = \int_{B_r} F_A(B_{\bar r}\mid x_r,r)~G(dx_r,r),
\end{aligned}
\end{equation}
where $B=B_r\otimes B_{\bar{r}}=\otimes_{j=1}^d B_j$, for measurable subsets $B_j$ of $\mcX_j$.   In this expression, $G(B_r,r)=\int_{B_r} g(x_r,r)\mu(dx_r)$ and $F_A(B_{\bar r}\mid x_r,r)=\int_{B_{\bar r}} f_A(x_{\bar r}\mid x_r,r)~\mu(dx_{\bar{r}})$. 

\subsection{Estimating the Full-Data Distribution} 	\label{sec::est::full}

To construct a nonparametric estimator of $F_A$, we propose to directly replace $G$ in \eqref{eq:FA} by the empirical observed-data distribution $\hat G$, 
\begin{equation}
\hat{G}(B_r,r) = \frac{1}{n}\sum_{i=1}^n I(X_{i,r}\in B_r, R_i = r).
\label{eq:Ghat}
\end{equation}

Handling  $F_A(B_{\bar r}\mid x_r,r)$ is more challenging, because identifying restrictions are typically expressed in terms of density functions, and so $f_A(x_{\bar r}\mid x_r,r)$ will usually have an explicit expression in terms of $g$, as seen in Section \ref{sec:PMMs}.  This means that we cannot typically obtain an estimate of $F_A(B_{\bar r}\mid x_r,r)$ directly as a function of the empirical observed-data distribution $\hat G$.
 We then propose to obtain a nonparametric estimate of $F_A(B_{\bar r}\mid x_r,r)$ via a \emph{surrogate estimate} of $g$, that is, an estimate of the observed-data density obtained for the sole purpose of being used in the estimation of $F_A(B_{\bar r}\mid x_r,r)$.  In particular, we take the surrogate as the kernel density estimator
\begin{equation}
\hat{g}_h(x_r, r)
= \frac{1}{n}\sum_{i=1}^{n} I(R_i = {r}) \prod_{j=1}^d K_j(x_j;X_{ij},h_j)^{r_j}.
\label{eq::kde}
\end{equation}
In this expression, $K_j$ represents a density with respect to an appropriate dominating measure, which depends on the nature of variable $X_j$, a location $X_{ij}$ and a smoothing parameter $h_j$.  For continuous variables, we take $K_j$ to be a Gaussian density centered at $X_{ij}$ with variance $h_j^2$, and so when all $d$ variables are continuous,  expression \eqref{eq::kde} closely corresponds to a traditional kernel density estimator \citep{Rosenblatt56,Parzen62,silverman1986density}.  For an unordered categorical variable with $C_j$ categories, we take $K_j(x_j;X_{ij},h_j) = I(x_j=X_{ij})h_j+I(x_j\neq X_{ij})(1-h_j)/(C_j-1)$, for $C_j^{-1}\leq h_j\leq 1$, which corresponds to the kernel method proposed by \cite{AitchisonAitken76}.  Different types of variables can be handled in this fashion as long as  appropriate kernels are used \citep{Titterington80,WangRyzin81,Kokonendji09,ChenTang11,LiRacine03}.  

We then obtain $\hat f_{A,h}(x_{\bar r}\mid x_r, r)$ using $\hat g_h(x_r,r)$ and the mapping induced by assumption $A$, that is, $\hat f_{A,h}=A(\hat g_h)$.  Then, the extrapolation distributions are obtained as
\begin{equation}
\hat{F}_{A,h}(B_{\bar r}\mid x_r,r) = \int_{B_{\bar r}}\hat{f}_{A,h}(x_{\bar r}\mid x_r, r)\mu(dx_{\bar r}).
\label{eq::extra}
\end{equation}
Finally, our estimator of the full-data distribution is obtained by plugging into equation \eqref{eq:FA}: 
\begin{equation}
\hat{F}_{A,h}(B,r)  = \frac{1}{n}\sum_{i=1}^n\hat{F}_{A,h}(B_{\bar r}\mid X_{i,r},r)
I(X_{i,r}\in B_r, R_i = r).
\label{eq::cdf}
\end{equation}

To illustrate this estimator we now provide details for monotone nonresponse, and provide the construction of $\hat{F}_{A,h}(B,r)$ under the CC restriction in Appendix \ref{sec::simple}. 

%
%

\begin{example}[Monotone missingness]\label{ex::monotone_FDD} As a concrete example, we present our estimation approach for the case of monotone missingness.  The nonmonotone case can be handled similarly under the identification approach presented in Section \ref{ss:ident_nonmonotone}, albeit with a more intricate notation.  The observed sample in this case is represented by $S_n=(X_{i,\leq T_i},T_i)_{i=1}^n$, where $X_{i,\leq T_i}=(X_{i,1},\dots,X_{i,T_i})$.  Following the notation in Section \ref{ss:ident_monotone}, we have that 
\begin{equation}\label{eq:est_product}
\hat{f}_{\assum,h}(x_{\bar r}\mid x_r, r) = \hat{f}_{\assum,h}(x_{>t}\mid x_{\leq t}, T = t)=\prod_{s=t+1}^{d} \hat{f}_{\assum,h}(x_{s}\mid x_{<s}, T= t),
\end{equation}
where 
\begin{equation}\label{eq:est_monotone_donor}
\hat{f}_{\assum,h}(x_{s}\mid x_{<s}, T= t)=\hat{g}_{h}(x_{s}\mid x_{<s}, T \in \assum_{ts}),
\end{equation}
and 
\begin{equation}\label{eq:gmixture}
\hat{g}_{h}(x_{s}\mid x_{<s}, T \in \assum_{ts})=\sum_{i=1}^{n}W_i(x_{<s};\assum_{ts})K_s(x_s;X_{is},h_s),
\end{equation}
with 
\begin{equation}\label{eq:W_i}
W_i(x_{<s};\assum_{ts})\propto I(T_i \in \assum_{ts}) \prod_{j=1}^{s-1} K_j(x_j;X_{ij},h_j).
\end{equation}
\end{example}

\begin{remark}[Finite-sample vs asymptotic nonparametric saturation] 
For an estimator of the full-data distribution $\hat{F}$, we say that it is
\emph{finite-sample nonparametrically saturated} if for any $r$ and any Borel measurable set $B_r\subseteq \mcX_{r}$, 
$\hat{F}(B_r\otimes \mcX_{\bar r}, r) = \hat G(B_r,r)$ for every sample size $n$. 
If $\hat{F}(B_r\otimes \mcX_{\bar r}, r) = \hat G(B_r,r)$ only when $n\rightarrow\infty$, $\hat{F}$
is called \emph{asymptotically nonparametrically saturated}. 
It is clear that the estimator in equation \eqref{eq::cdf} satisfies $\hat{F}_{A,h}(B_r\otimes\mcX_{\bar r},r)=\hat G(B_r,r)$
for all $n$, which means that $\hat{F}_{A,h}$ is finite-sample nonparametrically saturated. 
In Remark \ref{rm::ANS}, we give an example of an estimator which is only asymptotically nonparametrically saturated. 
\end{remark}

\subsection{Estimation of Functionals}	\label{sec::functional}

Typically, we are interested in estimating a statistical functional $\theta \equiv \theta(F)$.  We can construct an estimator by evaluating $\theta(\cdot)$ on \eqref{eq::cdf}, that is, 
\begin{equation}
\hat{\theta}_{A,h} \equiv \theta(\hat{F}_{A,h}).
\label{eq::par}
\end{equation}
We will show that under some conditions this estimator is consistent and has asymptotic normality (Corollary \ref{cor::linear}).  
While $\hat{\theta}_{A,h}$ may have a closed-form under some simple cases (see an example in Appendix \ref{sec::simple}),
evaluating $\theta(\cdot)$ on $\hat{F}_{A,h}$ is not easy in general. 
For example, computing moments of $X$ under the assumptions presented before involves computing integrals of ratios of kernel density estimates.  Here we propose a Monte Carlo approach to approximate $\hat{\theta}_{A,h}$, in particular under the large class of donor-based assumptions presented in Sections \ref{ss:ident_monotone} and \ref{ss:ident_nonmonotone}.

\subsubsection{Monte Carlo Approximation}\label{sec::MCE}

Under our estimate $\hat{F}_{A,h}$, the conditional distribution of $(X_{\bar r}\mid X_{r}=x_r, R=r)$ has a density $\hat f_{A,h}(x_{\bar r}\mid x_r, r)$, whereas the distribution of $(X_{R}, R)$ is discrete given by the empirical observed-data distribution.  We propose to approximate  $\theta(\hat{F}_{A,h})$ via Monte Carlo.  
In particular, we draw $X^{(v)}_{i,\bar{R}_i}$ from $\hat{F}_{A,h}(\cdot\mid X_{i,R_i},R_i)$, for $v=1,\dots,V$, for each observed sample point $(X_{i,R_i},R_i)$.  The result can be organized into $nV$ completed sample points $S_{n,V}=\{(X^{(v)}_i,R_i)_{i=1}^n\}_{v=1}^V$, where $X^{(v)}_i=(X_{i,R_i},X^{(v)}_{i,\bar R_i})$.  We can then define a Monte Carlo approximation of $\hat{F}_{A,h}$ as 
\begin{equation}
\hat{F}_{A,h}^{\textsc{mc}}(B,r)  = \frac{1}{nV}\sum_{i=1}^n I(X_{i,r}\in B_r, R_i = r)\sum_{v=1}^V I(X^{(v)}_{i,\bar r}\in B_{\bar r}),
\label{eq::cdfMC}
\end{equation}
and we approximate $\theta(\hat{F}_{A,h})$ by $\theta(\hat{F}_{A,h}^{\textsc{mc}})$.  Computing the latter is simple, as it entails evaluating the sample version of the functional $\theta(\cdot)$ on $S_{n,V}$.  For instance, if the parameter of interest is a  correlation coefficient, we can use the sample correlation coefficients computed on $S_{n,V}$.

We note that this Monte Carlo procedure is analogous to the one used in the multiple imputation approach of \cite{Rubin87}, in the sense that $S_{n,V}$ corresponds to the observed data $S_n$ being completed $V$ times.  Unlike in Rubin's multiple imputation, our goal here is to simply approximate $\theta(\hat{F}_{A,h})$, as we  assess the variability of $\theta(\hat{F}_{A,h})$ via the bootstrap, as explained in Section \ref{sec::bootstrap}.  


\begin{example}[Monotone missingness]	Continuing with Example \ref{ex::monotone_FDD}, we now need to draw $X^{(v)}_{i,> T_i}$ from $\hat{F}_{A,h}(\cdot\mid X_{i,\leq T_i},T_i)$, for $v=1,\dots,V$, and for each observed sample point $(X_{i,\leq T_i},T_i)$.  From \eqref{eq:est_product} we can see that a draw $X^{(v)}_{i,> T_i}$ can be obtained by sequentially sampling $X^{(v)}_{is}$, $s=T_i+1,\dots,d$, from a distribution with density $\hat{g}_{h}(x_{s}\mid X^{(v)}_{i,<s}, T \in \assum_{T_i,s})$, where $X^{(v)}_{i,<s}=(X^{(v)}_{i,<s-1},X^{(v)}_{i,s-1})$ and $X^{(v)}_{i,<T_i+1}=X_{i,\leq T_i}$.

Sampling from $\hat{g}_{h}(x_{s}\mid X^{(v)}_{i,<s}, T \in \assum_{T_i,s})$ can be accomplished using the mixture representation given by \eqref{eq:gmixture}, by sampling an index $\ell$ with probability $W_\ell(X^{(v)}_{i,<s};\assum_{T_i,s})$ as in \eqref{eq:W_i}, and then drawing $X^{(v)}_{is}$ from a distribution with density $K_s(x_s;X_{\ell s},h_s)$.  We summarize this procedure in Algorithm \ref{alg::MC}.
In the Online Supplementary Material we provide R code implementing Algorithm~\ref{alg::MC} for handling monotone missingness under donor-based identification.


\end{example}

\begin{algorithm}[tb]
\caption{Sampling for Monte Carlo Approximation under Monotone Missingness} 
\label{alg::MC}
\begin{algorithmic}
	\For{$v=1,\dots,V; i=1,\dots,n$}
		\State $t\gets T_i$, $X^{(v)}_{\leq t}\gets X_{i,\leq t}$
		\For{$s = t+1,\cdots, d$}
			\State Draw $\ell\in\{1,\cdots, n\}$ with probability $W_\ell(X^{(v)}_{<s};\assum_{ts})$
			\State Draw $X^{(v)}_{s}$ from the distribution with density $K_s(\cdot;X_{\ell, s},h_s)$
			\State $X^{(v)}_{\leq s} \gets (X^{(v)}_{<s}, X^{(v)}_s)$
		\EndFor
		\State $X^{(v)}_{i,> t} \gets X^{(v)}_{> t}$
	\EndFor\\
\Return{$\{(X_{i,\leq T_i},X^{(v)}_{i,> T_i})_{i=1}^n\}_{v=1}^V$}
\end{algorithmic}
\end{algorithm}


\begin{remark}[Alternative Monte Carlo approximation]	\label{rm::MCE2}
For each completed sample $S_{n}^{v} = \{(X_i^{(v)} ,R_i):i=1,\cdots, n\}$,
let $\hat{F}_{A,h}^{(v)}$ be its corresponding empirical distribution.
Then the average of the corresponding estimates 
$$
\tilde{\theta}_V = \frac{1}{V}\sum_{v=1}^V \theta\left(\hat{F}_{A,h}^{(v)}\right)
$$
provides another Monte Carlo approximation to $\theta(\hat{F}_{A,h})$.  This approach might be appealing if computing $\theta(\cdot)$ on $\hat{F}_{A,h}^{\textsc{mc}}$ is too computationally intensive.  In such case, each $\theta\left(\hat{F}_{A,h}^{(v)}\right)$ can be computed in parallel, and the results can be easily combined to obtain a final approximation $\tilde{\theta}_V$ of $\theta(\hat{F}_{A,h})$.  Nevertheless, we expect any potential bias of $\theta(\hat{F}_{A,h}^{\textsc{mc}})$ to be smaller than the bias of $\tilde{\theta}_V$, and therefore $\theta(\hat{F}_{A,h}^{\textsc{mc}})$ should be the preferred approximation of $\theta(\hat{F}_{A,h})$.
\end{remark}

\subsubsection{Parameters Defined by Estimating Equations}


If the parameter of interest $\theta$ is determined by an 
estimating function $\E(\xi(X,\theta))=0$ for some function $\xi$,
then our method can be viewed as a generalization of the work of \cite{WangChen09}, who consider the problem of estimating parameters defined by estimating functions when a continuous response variable is subject to missingness but covariates are fully observed. 
\cite{WangChen09} proposed a Monte Carlo method to impute the missing data,
%
and apply the empirical likelihood approach
to find parameter of interest. 
When only one continuous variable is subject to missingness and we set the bandwidth $h_s=0$ in 
Algorithm \ref{alg::MC}, our method will be the same as the Monte Carlo method of \cite{WangChen09}.

Note that \cite{WangChen09} also discussed how to use
the bootstrap for constructing a confidence interval for the parameter of interest. 
We will also analyze the same idea 
%
in Section~\ref{sec::bootstrap} 
with a more general setting in terms of both the missingness assumption
and the parameter of interest, i.e., we do not restrict ourselves to estimators from empirical likelihood methods. 

%
%
%
%
%
%
%
%
%
%

\subsection{Estimating the Full-Data Density}

If the goal is to recover the joint density function of the study variables, one possible estimator would be the one that uses the surrogate estimate of $g$, that is, $\hat{f}_{A,h}(x)=\sum_r \hat{f}_{A,h}(x,r)$, where 
\begin{equation}
\hat{f}_{A,h}(x,r)= \hat{g}_h(x_r,r)\hat{f}_{A,h}(x_{\bar r}\mid x_r,r).
\label{eq::est}
\end{equation}
Note that the distribution generated by $\hat{f}_{A,h}$ is not the same
as our proposed $\hat{F}_{A,h}$,
because the latter does not involve
any smoothing of the observed-data distribution, whereas the surrogate estimate $\hat{g}_h(x_r,r)$ does involves smoothing.

\begin{remark}
\label{rm::ANS}
As is mentioned before, the estimator $\hat{F}_{A,h}$ is
finite-sample nonparametrically saturated.  
However, the distribution induced by the density estimator $\hat{f}_{A,h}(x,r)$
is not because of the effect of smoothing in the construction of the surrogate estimator $\hat{g}_h(x_r,r)$.
The distribution induced by $\hat{f}_{A,h}(x,r)$
is nevertheless \emph{asymptotically nonparametrically saturated} 
under a good choice of smoothing bandwidths.  For example for continuous variables, a good choice of smoothing bandwidth $h=h_n\rightarrow 0$ leads $\hat{g}_h$
to be a consistent estimator of the true observed-data density function $g$. 
\end{remark}

\section{Bootstrap Confidence Intervals}	\label{sec::bootstrap}

In what follows, we discuss how to use the bootstrap approach to obtain inferences on functions of 
the full-data distribution. 
In particular, we use the empirical bootstrap \citep{efron1979,efron1994introduction}.  Let $\theta= \theta(F)$ be the parameter of interest and $\hat{\theta}_{A,h} = \theta(\hat{F}_{A,h})$ be our estimate. 
Let $S^{*}_n = \{(X^*_{1,R_1}, R^*_1),\cdots, (X^*_{n,R_n}, R^*_n)\}$ 
be a bootstrap sample obtained by sampling with replacement from $S_n$, and let $\hat\theta^*_{A,h} = \theta(\hat{F}_{A,h}^*)$ be the corresponding estimator, 
where $\hat{F}_{A,h}^*$ is constructed using equation \eqref{eq::cdf} with the bootstrap sample.
After repeating the bootstrap procedure $B$ times 
we obtain $\hat{\theta}^{*(1)}_{A,h},\cdots,\hat{\theta}^{*(B)}_{A,h}$,
$B$ bootstrap estimates of the parameter of interest.
We use the upper and lower $\alpha/2$ quantiles of these $B$ numbers
as the confidence interval of $\theta$. 
Specifically,
we propose to use the interval 
\[\hat{C}_{n,\alpha}=\left[\ell_\alpha, u_\alpha\right] = \left[\hat{\Gamma}^{-1}(\alpha/2), \hat{\Gamma}^{-1}(1-\alpha/2)\right],\]
where 
\[\hat{\Gamma}(t) = \frac{1}{B}\sum_{b=1}^B I(\hat{\theta}^{*(b)}_{A,h} \leq t),\]
as the confidence interval of $\theta$.
In Theorem \ref{thm::LBT}, we prove that $\hat{C}_{n,\alpha}$ is asymptotically
valid under appropriate assumptions.

Note that the above method is called the bootstrap percentile approach \citep{efron1994introduction,hall2013bootstrap}. 
There are other possible approaches for constructing a confidence interval, such as
using the bootstrap variance estimator to construct a normal confidence interval,
or bootstrapping the $t$-distribution. 
See \cite{efron1994introduction} and \cite{hall2013bootstrap} for more details on other approaches. 

When our estimator is constructed using the Monte Carlo approximation described in 
Section~\ref{sec::MCE}, for each bootstrap sample $S^*_n$, we apply Algorithm \ref{alg::MC} 
to obtain a completed sample $S^*_{n,V}$.
Using $S^*_{n,V}$ and its corresponding estimate of the full-data distribution $\hat{F}^{\textsc{mc}*}_{A,h}$, 
we then obtain our bootstrap estimate $\theta(\hat{F}^{\textsc{mc}*}_{A,h})$.  
Repeating this procedure for $B$ bootstrap samples leads to $\theta(\hat{F}^{\textsc{mc}*(1)}_{A,h}), \dots \theta(\hat{F}^{\textsc{mc}*(B)}_{A,h})$, from which we obtain the quantiles to construct the confidence interval $\hat{C}_{n,\alpha}$.

\begin{remark}
One can also use the bootstrap to infer the density function.
This is often done by bootstrapping the $L_\infty$ error of the density estimate. 
Specifically, we use the quantile of $\{\sup_x|\hat{	f}_{A,h}^{*(b)}(x,r) - \hat{f}_{A,h}(x,r)|\}_{b=1}^B$
as the width of the confidence band.
It can be shown that this leads to a valid confidence band for $f_{A,h}(x,r)$
for each $r$.
For more details, we refer to \cite{chernozhukov2014anti}
and \cite{chen2017nonparametric}. 
\end{remark}

\section{Theory}	\label{sec::theory}

In the theoretical analysis, we will show that the estimator $\hat{F}_{A,h}$
is a consistent estimator of $F$ under appropriate assumptions. 
To simplify the analysis, we assume that all variables in $X$ are continuous
and we will focus on the case where the sets $B_j = (-\infty, x_j]$
for some $x_j$
so that we can simply write $F(B,r) = F(x,r)$
for $B= \otimes_{j=1}^d B_j = \otimes_{j=1}^d (-\infty,x_j]$.
Our results, with the exception of the ones involving estimating densities, can be generalized
to the case with categorical variables very easily. 
Also, to simplify the problem,
we assume that the smoothing bandwidth of each variable is the same, i.e., $h_s=h$ for all $s=1,\cdots,d$,
and all kernel functions are the same, i.e., $K_s = K$ for all $s=1,\cdots, d$.
We will prove theory for the monotone missingness case, 
but the general case in Section \ref{ss:ident_nonmonotone}
can be derived in a similar way, albeit with a more complicated notation. 
With monotone missingness, it is easier to use the notation $T=|R|$, $t=|r|\in \{0,1,\cdots, d\}$ 
to represent the missingness pattern, 
so the joint distribution function is written as $F(x,t)$
and its estimator is $\hat{F}_{A,h}(x,t)$.  We note that $F(x,t)$ will denote the joint distribution function of study variables $X$
with $T=t$, namely,
\begin{equation}
F(x,t) = F(x\mid t) g(t) = \int_{-\infty}^x f(x',t)\mu(dx').
\end{equation}
An important note is that in the definition of $F(x,t)$ we are integrating over the sample space of $X$ until $x$, but not over the sample space of $T$, that is, $F(\infty,t)$ gives us the probability of $T=t$, $g(t)$.  Before moving forward, we define two asymptotic functions related to $\hat{F}_{A,h}(x,t)$ and $\hat{f}_{A,h}(x,t)$ as 
\begin{equation}
\begin{aligned}
\bar{F}_{A,h}(x,t)  
& = \int_{-\infty}^{x_{\leq t}}\bar{F}_{A,h}(x_{>t}\mid x_{\leq t}',t)G(dx_{\leq t}', t), \\
\bar{f}_{A,h}(x,t) &=  {g}(t)\bar {g}_h(x_{\leq t}\mid t)\bar{f}_{A,h}(x_{>t}\mid x_{\leq t},t),
\end{aligned}
\label{eq::exp}
\end{equation}
where 
\begin{equation}
\bar{F}_{A,h}(x_{>t}\mid x_{\leq t},t) = \int_{-\infty}^{x_{>t}} \bar{f}_{A,h}(x'_{> t}\mid x_{\leq t},t)\mu(dx'_{>t})
\label{eq::extra2}
\end{equation}
is the extrapolation distribution,
and 
$
\bar{f}_{A,h}(x_{> t}\mid x_{\leq t},t) = A (\bar{g}_h)
$
is the extrapolation density constructed by 
the smoothed density function
 $\bar{g}_h(x_{\leq t}\mid t) = \E(\hat{g}_h(x_{\leq t}\mid t))$
and $\bar{g}_h(x_{\leq t},t) = \bar{g}_h(x_{\leq t}\mid t) g(t)$.
Essentially, $\bar{f}_{A,h}(x,t)$
is constructed using the expected value of the kernel density estimator (KDE).
However, $\bar{f}_{A,h}(x,t)\neq \E(\hat{f}_{A,h}(x,t))$ although 
we do have $\bar{f}_{A,h}(x,t)= \E(\hat{f}_{A,h}(x,t)) + o(1)$ when $n\rightarrow\infty$ and $h\rightarrow0$. 
Later we will see that $\bar{f}_{A,h}(x,t)$ plays a central role in estimation theory and
asymptotic normality. 

Note that the density function corresponding to $\bar{F}_{A,h}(x,t)$ is not $\bar{f}_{A,h}(x,t)$, 
since the marginal distribution function $\bar{F}_{A,h}(x_{\leq t}, t) = G(x_{\leq t},t)$, which does not have $\bar{g}_h(x_{\leq t},t)$ as a density.  If we are thinking about the conditional distribution,
the density corresponding to
$\bar{F}_{A,h}(x_{>t}\mid x_{\leq t}, t)$ is $\bar{f}_{A,h}(x_{>t}\mid x_{\leq t}, t)$, and therefore $\bar{F}_{A,h}(x,t)$ and $\bar{f}_{A,h}(x,t)$ agree only on the extrapolation densities.

%
%
%
%

\subsection{Estimation Theory}

Let $\mathbf{UBC}_2$ denote the collection of
all functions with uniformly bounded second derivatives. 
Moreover, we use the simplified notation
\begin{equation}
\begin{aligned}
K_h\left(x;X_i\right)&\equiv \prod_{j=1}^{d} K_h\left(x_j;X_{ij}\right)\equiv \prod_{j=1}^{d} K_j(x_j;X_{ij},h),\\
K_h\left(x_{\leq t};X_{i,\leq t}\right)&\equiv \prod_{j=1}^{t} K_h\left(x_j; X_{ij}\right)\equiv \prod_{j=1}^{t} K_j(x_j;X_{ij},h).  
\end{aligned}
\label{eq::kernel}
\end{equation}
For univariate variables $w,z$, 
the kernel function $K_h\left(w;z\right)$ can be written as 
$$
K_h\left(w;z\right) = \frac{1}{h} K\left(\frac{w-z}{h}\right),
$$
where $K\left(z\right)$ is the conventional kernel function. For instance, 
the Gaussian kernel has the form $K(z) = \frac{1}{\sqrt{2\pi}}e^{-z^2/2}$.


{\bf Assumptions.}
\begin{itemize}
\item[(A1)] The true full-data distribution function $F(x,t)$ has a density function $f_0(x,t)$ satisfying
\begin{enumerate}
\item $\inf_{x\in\mathcal{X}} f_0(x,t)>0$ for each $t=1,\cdots, d$.
\item $f_0(x,t)\in \mathbf{UBC}_2$ for each $t=1,\cdots, d$.
\end{enumerate}

\item[(A2)] The statistical functional $\theta$ is Hadamard differentiable. 
\item[(K1)] $K(z)$ has at least second-order bounded derivative
and
$$
\int z^2 K(z) \mu(dz) <\infty, \qquad \int  K^2(z) \mu(dz) <\infty.
$$

\item[(K2)] 
Let $\mathcal{K} = \left\{z\mapsto K\left(\frac{z-w}{h}\right): w\in\mathbb{R}, \bar{h}>h>0\right\}$, 
for some fixed constant $\bar{h}.$
We assume that $\mathcal{K}$ is a VC-type class. Namely,
there exists constants $A,v$ and a constant envelope $b_0$ such that
\begin{equation}
\sup_{Q} N(\mathcal{K}, \mathcal{L}^2(Q), b_0\epsilon)\leq \left(\frac{A}{\epsilon}\right)^v,
\label{eq::VC}
\end{equation}
where $N(T,d_T,\epsilon)$ is the $\epsilon$-covering number for a
semi-metric set $T$ with metric $d_T$, and $\mathcal{L}^2(Q)$ is the $L_2$ norm
with respect to the probability measure $Q$.
\end{itemize}

Assumption (A1) is to ensure that the probability density function (PDF) of the true full-data distribution 
is bounded away from $0$ on the support $\mathcal{X}$. 
Note that this implies that the corresponding observed-data PDF $g_0$
has a density value uniformly bounded away from $0$.
We need bounded second derivatives so that the bias of the KDE will be 
of the order $O(h^2)$. 
Assumption (A2) ensures the smoothness of the parameter of interest. 
Most common statistical functions, such as quantiles and moments, satisfy this condition. 

Assumption (K1) is a common assumption for the kernel function so that
we can obtain the conventional rate of the bias and variance \citep{scott2015multivariate,wasserman2006all}. 
Assumption (K2) is an assumption to guarantee uniform convergence 
of a KDE \citep{Gine2002,Einmahl2005}. 
Most common kernel functions, such as the Gaussian kernel or a compact support kernel, 
satisfy this condition \citep{Gine2002}. 
Note that under assumption (K2), 
the requirements of the smoothing bandwidth $\frac{\log n}{nh^d}\rightarrow 0$ and $h\rightarrow 0$ 
amount to the uniform convergence of the KDE, see, e.g., 
\cite{Rinaldo2010, genovese2014nonparametric,chen2016generalized}.


For an identifying assumption $A$, we define
\begin{align*}
F_A(x,t) &= \int_{-\infty}^{x} f_A(x',t)\mu(dx'),\\
f_A(x,t) &= g(t) g(x_{\leq t}\mid t)f_A(x_{>t}\mid x_{\leq t},t). 
\end{align*}
Note the function $f_A$ is the corresponding PDF of $F_A = {A}(G)$.
If the identifying assumption is correct, we have $f_A(x_{> t}\mid x_{\leq t},t) =f(x_{>t}\mid x_{\leq t},t)$.
Namely, the extrapolation density based on the identifying assumption is the same as 
the true extrapolation density. 
This
further implies $f_A(x) = f (x)$ and $F_A(x) = F(x)$.
Because the extrapolation density $f_A(x_{> t}\mid x_{\leq t}, t)$ plays a crucial role in our analysis,
we first derive perturbation theory for $f_A(x_{> t}\mid x_{\leq t}, t)$. 


\begin{theorem}[Perturbation theory of pattern mixture model]
Assume (A1)
and consider 
a donor-based identifying restriction with an identifying set $\assum_{ts}$.
Let $\tilde{g}_n$
be a sequence of observed-data density functions 
that is close to $g$ in the sense that
\begin{equation}
\begin{aligned}
\Delta g(x_{\leq s}\mid T\in \assum_{ts})
&\equiv \tilde g_n(x_{\leq s}\mid T\in \assum_{ts}) -  g(x_{\leq s}\mid T\in \assum_{ts})\rightarrow 0,\\
\Delta g(x_{< s}\mid T\in \assum_{ts})
&\equiv \tilde g_n(x_{< s}\mid T\in \assum_{ts}) -  g(x_{< s}\mid T\in \assum_{ts})\rightarrow 0,\\
\Delta g(t) &\equiv \tilde{g}_n(t) - g(t)\rightarrow 0
\end{aligned}
\label{eq::limit}
\end{equation}
uniformly for all $x\in\mcX$ and $t=1,\cdots, d$. 
Define 
$$
\Delta {f}_A(x_{> t}\mid x_{\leq t}, t) \equiv \tilde {f}_A(x_{> t}\mid x_{\leq t}, t)- {f}_A(x_{> t}\mid x_{\leq t}, t).
$$
Then 
\begin{align*}
\frac{\Delta {f}_A(x_{> t}\mid x_{\leq t}, t) }{f_A(x_{> t}\mid x_{\leq t}, t)} &= 
\sum_{s=t+1}^d \left\{\frac{\Delta g(x_{\leq s}\mid T\in \assum_{ts})}{g(x_{\leq s}\mid T\in \assum_{ts})}
+ \frac{\Delta g(x_{< s}\mid T\in \assum_{ts})}{ g(x_{< s}\mid T\in \assum_{ts})}\right\}
+\tilde{W}_n(x,t),
\end{align*}
where $\sup_x \left|\frac{\tilde{W}_n(x,t)}{\Delta {f}_A(x_{> t}\mid x_{\leq t}, t)} \right|\rightarrow 0$.
\label{thm::pdf}
\end{theorem}

Theorem~\ref{thm::pdf} shows the linear perturbation theory for $f_{A}(x_{> t}\mid x_{\leq t},t)$:
when the observed-data distribution is slightly perturbed, 
the corresponding extrapolation distribution is also slightly perturbed,
and the perturbation in the extrapolation distribution is of a linear order
to the perturbation in the observed-data distribution.  
Using this theorem, we can derive the asymptotic theory for $\hat{F}_{A,h}$.

\begin{theorem}
Assume (A1), (K1-2), $h\rightarrow 0$ and $\frac{\log n}{nh^d}\rightarrow 0$.
Then for a pattern-mixture model with a donor-based identifying restriction $A$,
\begin{align*}
\hat{F}_{A,h}(x,t) - \bar {F}_{A,h}(x,t) &=  O_P\left(\sqrt{\frac{1}{n}}\right)\\
\bar {F}_{A,h}(x,t) -F_A(x,t) &=  O(h^2)
\end{align*}
for each $t=1,\cdots, d$.
Moreover, under the same assumptions,
\begin{align*}
\hat{f}_{A,h}(x,t) - \bar {f}_{A,h}(x,t) &=  O_P\left(\sqrt{\frac{1}{nh^d}}\right)\\
\bar {f}_{A,h}(x,t) -f_A(x,t) &=  O(h^2)
\end{align*}
for each $t=1,\cdots, d$.
Thus, if the identifying assumption $A$ is correct, then we immediately have
$$
\hat{F}_{A,h}(x,t) - F(x,t) =  O(h^2) +O_P\left(\sqrt{\frac{1}{n}}\right), 
\quad\hat{f}_{A,h}(x,t) - {f}(x,t) =  O(h^2) + O_P\left(\sqrt{\frac{1}{nh^d}}\right)
$$
for each $t=1,\cdots, d$.

\label{thm::kde}
\end{theorem}

The first part of Theorem~\ref{thm::kde} describes the convergence rate
of $\hat{F}_{A,h}$. 
Unlike the case of estimating the density function, estimating
the cumulative density function (CDF) has a $\sqrt{n}$ convergence rate even if it is built by integrating
a KDE. This result is known in the literature, see, e.g., \cite{reiss1981nonparametric,liu2008kernel}.


The estimation rate in
Theorem~\ref{thm::kde} implicitly uses a decomposition of the uncertainty:
\begin{equation}
\begin{aligned}
\hat{F}_{A,h}(x,t) - F(x,t) = \underbrace{\hat{F}_{A,h}(x,t) - \bar {F}_{A,h}(x,t)}_\text{Stochastic error} +\underbrace{\bar F_{A,h}(x,t) - F_A(x,t)}_\text{Smoothing/model bias} + \underbrace{F_A(x,t) -F(x,t)}_\text{Restriction bias}.
\end{aligned}
\end{equation}
The stochastic error comes from the sampling variability. 
Confidence intervals are often designed to capture this type of uncertainty 
(later we will see this more explicitly).
The smoothing/model bias comes from the bias of our model. 
When using a KDE, it is of the order $O(h^2)$. 
We can reduce its effect by choosing a small smoothing bandwidth (undersmoothing).
The restriction bias comes from misspecification of the identifying restriction. 

With the convergence rate of $\hat{F}_{A,h}$, 
the rate for estimating $\theta_{A,h} = \theta(F_{A,h})$
is 
$$
\hat{\theta}_{A,h} - \theta_{A,h} = O_P\left(\sqrt{\frac{1}{n}}\right),
$$
when $\theta$ is a smooth functional (such as being Hadamard differentiable, e.g., \citealt{van1998asymptotic}). 

In addition to the convergence rate, the estimated distribution function $\hat{F}_{A,h}$
also has a beautiful asymptotic behavior as illustrated in the following theorem.

\begin{theorem}
Assume (A1), (K1-2), $h\rightarrow 0$ and $\frac{\log n}{nh^d}\rightarrow 0$.
Then for each $t=1,\cdots, d$,
$\sqrt{n}\left(\hat{F}_{A,h}(\cdot, t) - \bar{F}_{A,h}(\cdot, t)\right)$
converges weakly to a Brownian bridge in $L_\infty$ norm.
\label{thm::uclt}
\end{theorem}

Theorem~\ref{thm::uclt} shows that the difference between the estimated CDF $\hat{F}_{A,h}$
and its population counterpart $\bar{F}_{A,h}$
forms a Brownian bridge.
This implies several powerful results. 
For instance, for any given point $x$ and pattern $t$, $\hat{F}_{A,h}(x,t) - \bar{F}_{A,h}(x,t)$
has asymptotic normality. 
For any Hadamard differentiable functional $\theta(\cdot)$,
$\hat{\theta}_{A,h}-\bar\theta_{A,h}$ also has asymptotic normality (see Corollary \ref{cor::linear}).

We now provide a high-level idea of the proof of Theorem~\ref{thm::uclt}.
Recall that the CDF estimator can be written as an integral of the PDF estimator. 
Because of Theorem~\ref{thm::pdf}, the difference $\hat{F}_{A,h} - \bar{F}_{A,h}$
can be rewritten as
the difference between an integrated KDE and a CDF.
Then, we apply the uniform central limit theorem of a smoothed empirical process 
(the main theorem of \cite{van1994weak} or Theorem 2 of \cite{gine2008uniform}),
which implies that the difference between an integrated KDE and its expectation (a CDF)
converges uniformly to a Brownian bridge
and the result follows.



\subsection{Theory of Monte Carlo Estimate}

Now we show that the Monte Carlo approach of Algorithm \ref{alg::MC} indeed 
generates points from the estimated full-data distribution. 

\begin{theorem}
The observations drawn using Algorithm~\ref{alg::MC} are such that 
$$
X_{i,>T_i}^{(1)},\cdots,X_{i,>T_i}^{(V)} \overset{IID}{\sim} \hat{F}_{A,h}(x_{>T_i}\mid X_{i,\leq T_i}, T_i).
$$

\label{thm::MC_icin}
\end{theorem}

Here is an intuitive explanation of why Theorem~\ref{thm::MC_icin}
works. 
Equations \eqref{eq:est_monotone_donor}
and \eqref{eq:gmixture}
imply that given $T=t$ and any $s>t$, 
the extrapolation density 
$$
\hat{f}_{A,h}(x_s\mid x_{<s}, T=t)= \hat{g}_h(x_s\mid x_{<s}, T\in\assum_{ts})
= \sum_{i=1}^n W_i(x_{<s};\assum_{ts}) K_s(x_s;X_{is}, h_s),
$$
which can be viewed as a density mixture such that 
with a probability of  $W_i(x_{<s};\assum_{ts})$,
we will sample from the density $K_s(x_s;X_{is}, h_s)$. 
Essentially, Algorithm~\ref{alg::MC} is following this feature of a density mixture
for each $s=t+1,t+2,\cdots, d$,
so the imputed observations are IID from the desired distribution.

\subsection{Bootstrap Theory}

We now discuss the validity of bootstrap confidence bands and intervals. 
We first introduce a corollary showing that the estimator
$\hat{\theta}_{A,h}$ is asymptotically normal.



%

\begin{corollary}
Assume (A1-2), (K1-2), and $h\rightarrow0, \frac{\log n}{nh^d}\rightarrow 0$. 
Then 
$$
\sqrt{n}(\hat{\theta}_{A ,h} - \bar \theta_{A,h}) \overset{D}{\rightarrow} N(0,\sigma^2),
$$
for a constant $\sigma^2>0$.
\label{cor::linear}
\end{corollary}

This corollary follows from the functional delta method and Theorem~\ref{thm::uclt},
so we omit the proof. 
The crucial assumption here is that the statistical functional $\theta(\cdot)$
is Hadamard differentiable. 
This assumption holds for many common statistical functionals
such as the mean, median, variance, correlation between two variables, etc.
The formal definitions and more details are given in Appendix \ref{sec::Hada}.

With Corollary~\ref{cor::linear},
we derive the validity of the bootstrap confidence interval as follows. 

\begin{theorem}
Under the assumptions (A1-2) and $h\rightarrow0, \frac{\log n}{nh^d}\rightarrow 0$,
$$
\sqrt{n}(\hat{\theta}^*_{A,h} - \hat\theta_{A,h})\overset{D}{\rightarrow} N(0,\sigma^2)
$$
and
$$
\sup_q\left|P\left(\sqrt{n}(\hat{\theta}^*_{A,h} - \hat\theta_{A,h})<q|S_n\right) -P\left(\sqrt{n}(\hat{\theta}_{A ,h} - \bar\theta_{A,h})<q\right)\right| 
\overset{P}{\rightarrow}0.
$$
Thus,
$$
P(\bar \theta_{A,h}\in\hat{C}_{n,\alpha}) = 1-\alpha+o(1).
$$

\label{thm::LBT}
\end{theorem}

With the convergence towards a Brownian bridge (Theorem~\ref{thm::uclt}),
this result follows from the Theorem of the bootstrap for the delta method; see, e.g., Theorem 23.9 of \cite{van1998asymptotic}
and Theorem 3.9.11 of \cite{van1996weak}.
Thus, we omit the proof.

Theorem~\ref{thm::LBT} shows that the bootstrap method can recover
the uncertainty of $\hat\theta_{A,h}$, which further leads to the validity
of a bootstrap confidence interval. 
An advantage of using the bootstrap is that there is no need to calculate 
$\sigma^2$. 
If $\sigma^2$ has a closed form and can be estimated, 
we can use the estimated $\sigma^2$ to construct a normal confidence interval
or use the bootstrap $t$-distribution \citep{hall2013bootstrap} to construct another bootstrap confidence interval. 
Even when we do not know $\sigma^2$, we can use the sample variance of the bootstrap estimates
to estimate this quantity and use it to construct a confidence interval.

\subsection{Bootstrap Diagram}	\label{sec::BD}


\begin{figure}[h]
\center
\includegraphics[width=.7\linewidth]{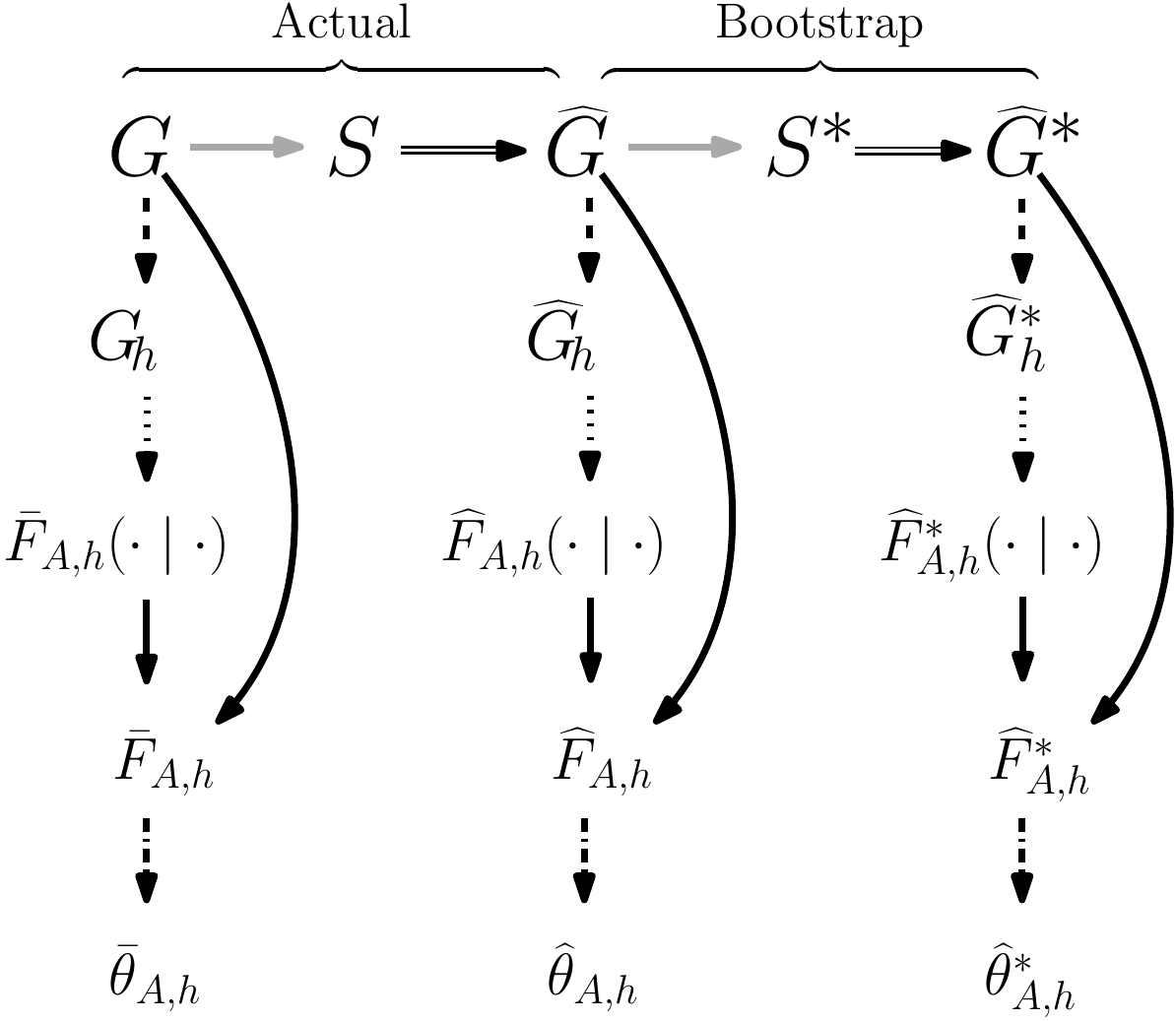}
\caption{Extended Efron's bootstrap diagram.  The following concepts are represented by the arrow types in parenthesis: sampling (solid gray), modeling/density estimation (dashed black), 
extrapolating (dotted black), reconstructing CDF (solid black),
and derivation of parameter of interest (dashed dotted black).}
\label{fig:diagram}
\end{figure}

In Theorem~\ref{thm::LBT}, we see that the bootstrap provides a valid confidence interval
of $\bar{\theta}_{A,h} = \theta(\bar{F}_{A,h})$, the parameter of interest corresponding
to $\bar{F}_{A,h}$ in \eqref{eq::exp}. We now explain why this is the underlying population quantity
 using the concept of \emph{bootstrap diagram} from \cite{Efron_missingdata}.  Figure \ref{fig:diagram} shows an expanded version of Efron's bootstrap diagram to explicitly illustrate the effect of modeling (or smoothing) and the role of identifying assumptions.  Initially we shall think of $\hat G$ as the empirical distribution based on the observed sample $S$, and the observed sample is from the true distribution function $G$.
During our construction of the estimator $\hat{F}_{A,h}$,
we first apply smoothing to $\hat{G}$ to obtain the KDE $\hat{g}_h$
and then combine these densities using the identifying assumption $A$ to obtain an estimate
of the extrapolation distribution
$\hat{F}_{A,h}(\cdot\mid\cdot)$ defined in equation \eqref{eq::extra}.
Thus, we can view the kernel smoothing as a mapping from $\hat{G}$ to $\hat{G}_{h}$,
where $\hat{G}_{h}$ denotes the CDF generated by $\hat{g}_h$. 
%
The estimator of the joint CDF $\hat{F}_{A,h}$ in equation \eqref{eq::cdf}
is constructed using a mapping with two inputs: the extrapolation distribution $\hat{F}_{A,h}(\cdot\mid\cdot)$
and the empirical distribution function $\hat{G}$. 
Therefore, in the bootstrap diagram we use two arrows ($\hat{F}_{A,h}(\cdot\mid\cdot)\rightarrow \hat{F}_{A,h}$
and $\hat{G}\rightarrow\hat{F}_{A,h}$) to denote this. 
The estimate of the parameter of interest $\hat{\theta}_{A,h} = \theta(\hat{F}_{A,h})$
is simply a plug-in estimate so it can be viewed as a mapping from $\hat{F}_{A,h}$
to $\hat{\theta}_{A,h}$.

In the bootstrap process, we are generating observations from the empirical distribution function $\hat{G}$,
so $\hat{G}$ now serves the role of $G$. 
Therefore, the bootstrap sample $S^*$ generates another empirical distribution function $\hat{G}^*$,
which also leads to $\hat{g}^*_h$, the bootstrap estimates of 
the density functions of the observed-data distribution. 
Just like the case of original estimate, $\hat{G}^*_{h}$ is the CDF corresponds
to $\hat{g}^*_h$
and the identifying assumption maps it into the extrapolation distribution $\hat{F}^*_{A,h}(\cdot\mid\cdot)$. 
The extrapolation distribution $\hat{F}^*_{A,h}(\cdot\mid\cdot)$,
together with the bootstrap empirical distribution $\hat{G}^*$,
leads to the CDF estimator $\hat{F}^*_{A,h}$, 
and the bootstrap estimate of the parameter of interest 
$\hat{\theta}^*_{A,h} = \theta(\hat{F}^*_{A,h})$.


So far, we have explained the middle and the right-hand branches of the bootstrap diagram in Figure \ref{fig:diagram}. 
We will now explain the left-hand branch, which will also clarify why the population quantity that
the confidence interval covers is $\bar{\theta}_{A,h} = \theta(\bar{F}_{A,h})$. 
The key step is to notice that there is a mapping between $\hat{G}$ to $\hat{G}_{h}$
and $\hat{G}^*$ to $\hat{G}^*_{h}$ that represents the effect
of smoothing/modeling.  
If we start with $G$, we also need to define a quantity $G_{h}$
that is a quantity based on smoothing $G$.
Thus, it is not hard to see that $G_{h}$ denotes the CDF corresponding
to the smoothed densities $g_h = \E(\hat{g}_h)$.
Then the identifying assumption $A$ maps $G_{h}$ to an extrapolation distribution $\bar{F}_{A,h}(\cdot\mid\cdot)$
defined in equation \eqref{eq::extra2}. 
The quantity $\bar{F}_{A,h}$ in equation \eqref{eq::exp}
is constructed using the extrapolation distribution $\bar{F}_{A,h}(\cdot\mid\cdot)$ and
the observed-data distribution $G$
so again we have two arrows from $\bar{F}_{A,h}(\cdot\mid\cdot)$ and $G$
to $\bar{F}_{A,h}$. 
With the full-data distribution $\bar{F}_{A,h}$, we can view the corresponding parameter of interest,
$\bar{\theta}_{A,h} = \theta(\bar F_{A,h}) $,
as a mapping from $\bar{F}_{A,h}$ to $\bar{\theta}_{A,h} ,$
which is represented by the dotted dashed arrow at the bottom of Figure \ref{fig:diagram}. 
It is now easy to see why $\bar{\theta}_{A,h}$ 
is the parameter of interest that the bootstrap confidence interval is covering -- the bootstrap 
difference
$\hat{\theta}^*_{A,h} -\hat{\theta}_{A,h}$
is mimicking the difference
$\hat{\theta}_{A,h} -\bar{\theta}_{A,h}$. 

Although Figure~\ref{fig:diagram} displays the bootstrap diagram
when we are using the KDE in the modeling stage, 
this diagram can be generalized to an arbitrary density estimator or modeling method. 
If we use a model $\mathcal{M}$ that maps $\hat G$ into a density estimator $\hat{g}_{\mathcal{M}}$
with a CDF estimator $\hat{G}_{\mathcal{M}}$,
the diagram remains the same except that we replace every element with a subscript $h$ 
by a subscript $\mathcal{M}$.
For example, one could model the observed-data density functions using Gaussians as in \cite{Little93}, in which case
 each $\hat{g}_{\mathcal{M}}(x_{ r}\mid r)$ will correspond to a fitted Gaussian PDF, and $\hat{g}_{\mathcal{M}}(r)$ could simply correspond to the empirical frequency of pattern $r$.  With these components, we can construct the extrapolation distribution $\hat{F}_{A,\mathcal{M}}(\cdot\mid\cdot)$, and with these elements a joint full-data distribution
can be constructed in a similar manner as equation \eqref{eq::cdf}.

\section{Data Analysis}	\label{sec::data}

To illustrate the usage of nonparametric pattern-mixture models, we now present an application to the analysis of data coming from a clinical trial on schizophrenia.  These data had been previously analyzed under parametric or semiparametric longitudinal models \citep[e.g.,][]{Diggleetal07}.  The purpose of the trial was to evaluate the effectiveness of four different doses of a new treatment (risperidone, an antipsychotic medication, with 2, 6, 10 or 16 mg/day) compared with placebo and with a standard of care (20 mg/day of haloperidol, a standard antipsychotic), in patients with chronic schizophrenia \citep{Risperidone94}.  The Positive and Negative Syndrome Scale for Schizophrenia (PANSS) score was measured on patients one week before, the day of, and on weeks 1, 2, 4, 6, and 8 after randomization.  In the left panel of Figure \ref{fig::summary} we summarize the frequency of dropout times for each arm of the trial.  Here for simplicity the new treatment is taken as the arm of 6 mg/day of risperidone, since this was found to be the most effective dose, and so we omit the results for the other risperidone arms.  The center and right panels of Figure \ref{fig::summary} present the observed means of the PANSS score among those who were last seen in week 4 and those who completed the study.  
\begin{figure}
\center
\includegraphics[width=6 in]{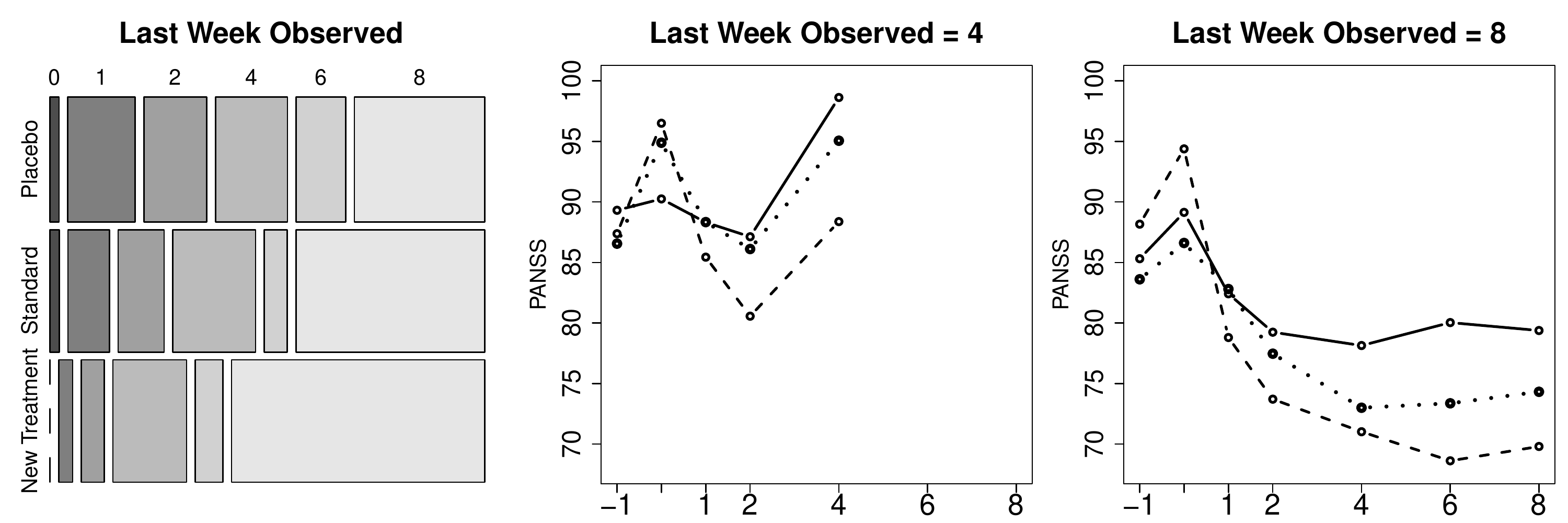}
\caption{Left: distribution of dropout per each treatment arm.  Center and right: observed means of PANSS score for two dropout times, where solid lines represent the placebo, dotted lines the standard treatment, and dashed lines the new treatment.}
\label{fig::summary}
\end{figure}

In this context we are interested in estimating average treatment effects (ATEs) over time.  Let $X_j$ denote the PANSS score at time $j$.  Denote $\mu_{j}^N=E(X_j\mid \text{New})$, $\mu_{j}^S=E(X_j\mid \text{Standard})$, and $\mu_{j}^P=E(X_j\mid \text{Placebo})$.  The ATEs that we are interested in are $\mu_j^N-\mu_j^P$, new treatment vs placebo; $\mu_j^S-\mu_j^P$, standard treatment vs placebo; and $\mu_{j}^N-\mu_{j}^S$, new vs standard treatment.  Since the true ATEs are not accessible to us, not even with infinite samples, we focus on estimating ATEs under a given identifying restriction $A$, and denote the corresponding means as $\mu^N_{j,A}$, $\mu^S_{j,A}$, $\mu^P_{j,A}$ to emphasize the dependence on the assumption.    

Our approach estimates the full-data distribution under an assumption $A$, approximates it via a Monte Carlo procedure, which is then used to evaluate functionals of interest.  To implement our methodology, we used Gaussian kernels in the construction of the surrogate estimate \eqref{eq::kde}, using Silverman's rule \citep{silverman1986density} to compute the bandwidths with  the observed PANSS scores for each week.  We then implemented our Monte Carlo approximation using Algorithm \ref{alg::MC}, taking $V=100$ Monte Carlo samples.  We used the AC, 3NC and NC assumptions, as explained in Section \ref{ss:ident_monotone}. We note that the AC restriction is equivalent to the missing at random (MAR) assumption under mononone missingness, and in this case also equivalent to the 5NC restriction.  This approach thus provides us with a way of performing sensitivity analysis to the commonly used MAR assumption.   

We compute confidence intervals repeating our estimating procedure over 1000 bootstrap samples, as described in Section \ref{sec::bootstrap}.  In Figure \ref{fig::CIs} we present the point-wise bootstrap 95\% confidence intervals and point estimates of the ATEs.  Often the main interest is in estimating the ATE at the last time point.  Furthermore, one could argue that here the most interesting ATE is $\mu^N_{8,A}-\mu^S_{8,A}$, as it compares the new treatment with the standard of care at week 8.  From Figure \ref{fig::CIs} we can see that all the confidence intervals for this ATE fall below zero, under each of the three missing data assumptions considered here.  This gives us compelling evidence for declaring superiority of the new treatment over the standard of care, as this conclusion seems to be insensitive to the missing data assumption.  

\begin{figure}
\center
\includegraphics[width=6 in]{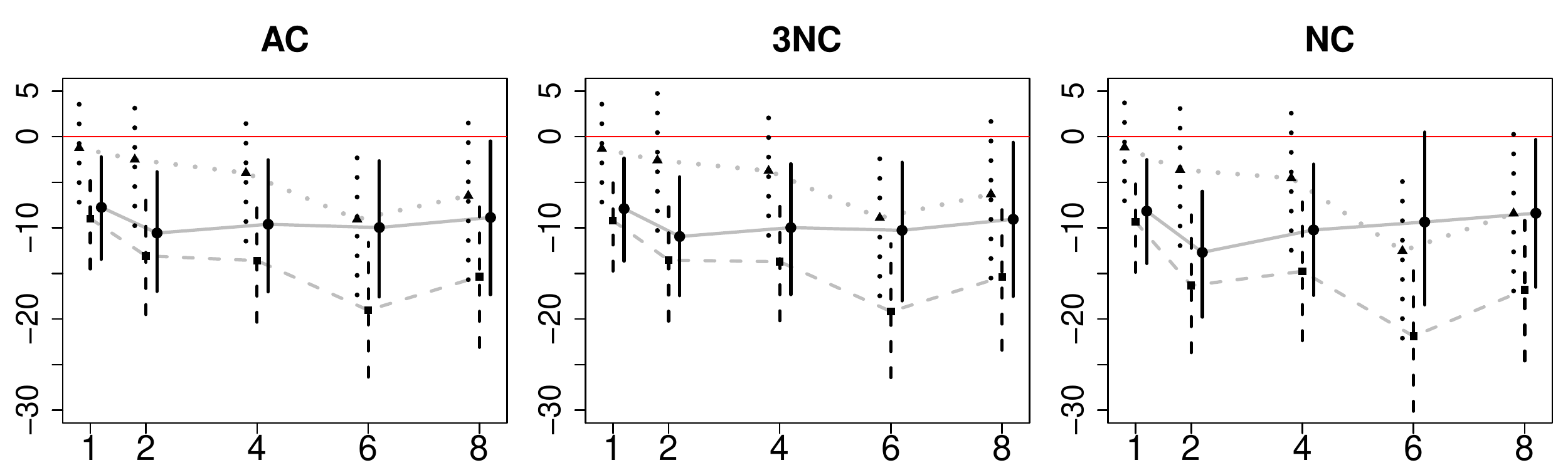}
\caption{Estimated ATEs from a clinical trial on schizophrenia under three different missing-data identifying assumptions.  Dashed lines indicate results for the ATEs that compare the new treatment vs the placebo, $\mu_{j,A}^N-\mu_{j,A}^P$; dotted lines compare the standard treatment vs placebo, $\mu_{j,A}^S-\mu_{j,A}^P$; and solid lines compare new vs standard treatment $\mu_{j,A}^N-\mu_{j,A}^S$; for times $j=1,2,4,6,8$, and assumptions $A=$AC, 3NC, NC.  Vertical lines represent point-wise 95\% bootstrap confidence intervals.}
\label{fig::CIs}
\end{figure}

\section{Discussion}	\label{sec::discussion}

We studied how to conduct nonparametric inference with missing data under the pattern-mixture model 
framework. 
We introduced the concept of donor-based identification under monotone missingness and generalized it
to nonmonotone missing data. 
We proposed an estimator of the full-data distribution based on a surrogate estimator of the observed-data distribution based on kernel smoothing, derived the corresponding convergence rates 
and proved asymptotic normality. 
To numerically compute the estimator of functionals of interest, 
we proposed a Monte Carlo method that 
allows us to easily sample from the estimated full-data distribution. 
We also introduced a bootstrap method for constructing confidence intervals, and we revisited Efron's bootstrap diagram to explain why and how the bootstrap method works.  We presented the underlying theory supporting our methodology, and numerical analyses to illustrate the applicability of our approach in practice. 

Our current results were developed under donor-based identification, which cover important identifying restrictions used in pattern-mixture models, but it is of interest to study extensions of our methods to handle other identifying restrictions that analysts might want to specify.  Two challenges arise when thinking of working with totally generic restrictions.  First, it is not clear how to design a Monte Carlo method for sampling from the estimated full-data distribution under arbitrary restrictions.  Second, perturbation theory similar to the one in Theorem~\ref{thm::pdf} needs to be derived, although once one has it, convergence rates, asymptotic normality, and the validity of the bootstrap can all be derived in a similar manner.

Another open problem is how to reduce the computational cost for the bootstrap method
when the estimator is approximated via Monte Carlo.  
Our current approach requires using Monte Carlo (Algorithm~\ref{alg::MC})
for every bootstrap sample. 
When the sample size is large, this procedure may not be computationally appealing, although one could parallelize it. 
An interesting observation is that some functionals under some restrictions might be obtainable in closed form, such as the mean under the CC restriction, as shown in Appendix \ref{sec::simple}, thereby avoiding the need for the Monte Carlo approximation. 

\section*{Acknowledgement}
YC is supported by NSF grant DMS 1810960 and NIH grant U01 AG016976.

\bibliographystyle{apa-good}
\bibliography{bootstrap}

\begin{thebibliography}{56}
\expandafter\ifx\csname natexlab\endcsname\relax\def\natexlab#1{#1}\fi
\expandafter\ifx\csname url\endcsname\relax
  \def\url#1{{\tt #1}}\fi
\expandafter\ifx\csname urlprefix\endcsname\relax\def\urlprefix{URL }\fi

\bibitem[{Aitchison \& Aitken(1976)}]{AitchisonAitken76}
Aitchison, J., \& Aitken, C. G.~G. (1976).
\newblock {Multivariate Binary Discrimination by the Kernel Method}.
\newblock {\em Biometrika\/}, {\em 63\/}(3), 413--420.

\bibitem[{Chen \& Tang(2011)}]{ChenTang11}
Chen, S.~X., \& Tang, C.~Y. (2011).
\newblock {Nonparametric regression with discrete covariate and missing
  values}.
\newblock {\em Statistics and Its Interface\/}, {\em 4\/}(4), 463--474.

\bibitem[{Chen(2016)}]{chen2016generalized}
Chen, Y.-C. (2016).
\newblock Generalized cluster trees and singular measures.
\newblock {\em arXiv preprint arXiv:1611.02762\/}.

\bibitem[{Chen(2017)}]{chen2017nonparametric}
Chen, Y.-C. (2017).
\newblock Nonparametric inference via bootstrapping the debiased estimator.
\newblock {\em arXiv preprint arXiv:1702.07027\/}.

\bibitem[{Chernozhukov et~al.(2014)Chernozhukov, Chetverikov, \&
  Kato}]{chernozhukov2014anti}
Chernozhukov, V., Chetverikov, D., \& Kato, K. (2014).
\newblock Anti-concentration and honest, adaptive confidence bands.
\newblock {\em The Annals of Statistics\/}, {\em 42\/}(5), 1787--1818.

\bibitem[{Daniels \& Hogan(2000)}]{DanielsHogan00}
Daniels, M., \& Hogan, J. (2000).
\newblock {Reparameterizing the pattern mixture model for sensitivity analysis
  under informative dropout in longitudinal studies}.
\newblock {\em Biometrics\/}, {\em 56\/}, 1241--1249.

\bibitem[{Daniels \& Hogan(2008)}]{DanielsHogan08}
Daniels, M.~J., \& Hogan, J.~W. (2008).
\newblock {\em {Missing Data in Longitudinal Studies: Strategies for Bayesian
  Modeling and Sensitivity Analysis}\/}.
\newblock Boca Raton: Chapman and Hall/CRC.

\bibitem[{Diggle et~al.(2007)Diggle, Farewell, \& Henderson}]{Diggleetal07}
Diggle, P., Farewell, D., \& Henderson, R. (2007).
\newblock {Analysis of longitudinal data with dropout: objectives, assumptions
  and a proposal}.
\newblock {\em Journal of the Royal Statistical Society: Series C (Applied
  Statistics)\/}, {\em 56\/}(5), 499--550.

\bibitem[{Dvoretzky et~al.(1956)Dvoretzky, Kiefer, \&
  Wolfowitz}]{dvoretzky1956asymptotic}
Dvoretzky, A., Kiefer, J., \& Wolfowitz, J. (1956).
\newblock Asymptotic minimax character of the sample distribution function and
  of the classical multinomial estimator.
\newblock {\em The Annals of Mathematical Statistics\/}, (pp. 642--669).

\bibitem[{Efron(1979)}]{efron1979}
Efron, B. (1979).
\newblock Bootstrap methods: Another look at the jackknife.
\newblock {\em Ann. Statist.\/}, {\em 7\/}(1), 1--26.
\newline\urlprefix\url{https://doi.org/10.1214/aos/1176344552}

\bibitem[{Efron(1994)}]{Efron_missingdata}
Efron, B. (1994).
\newblock Missing data, imputation, and the bootstrap.
\newblock {\em Journal of the American Statistical Association\/}, {\em
  89\/}(426), 463--475.

\bibitem[{Efron \& Tibshirani(1994)}]{efron1994introduction}
Efron, B., \& Tibshirani, R.~J. (1994).
\newblock {\em An introduction to the bootstrap\/}.
\newblock CRC press.

\bibitem[{Einmahl \& Mason(2005)}]{Einmahl2005}
Einmahl, U., \& Mason, D.~M. (2005).
\newblock Uniform in bandwidth consistency of kernel-type function estimators.
\newblock {\em The Annals of Statistics\/}, {\em 33\/}(3), 1380--1403.

\bibitem[{Fitzmaurice(2003)}]{Fitzmaurice03}
Fitzmaurice, G.~M. (2003).
\newblock Methods for handling dropouts in longitudinal clinical trials.
\newblock {\em Statistica Neerlandica\/}, {\em 57\/}(1), 75--99.

\bibitem[{Genovese et~al.(2014)Genovese, Perone-Pacifico, Verdinelli, \&
  Wasserman}]{genovese2014nonparametric}
Genovese, C.~R., Perone-Pacifico, M., Verdinelli, I., \& Wasserman, L. (2014).
\newblock Nonparametric ridge estimation.
\newblock {\em The Annals of Statistics\/}, {\em 42\/}(4), 1511--1545.

\bibitem[{Gin{\'e} \& Guillou(2002)}]{Gine2002}
Gin{\'e}, E., \& Guillou, A. (2002).
\newblock Rates of strong uniform consistency for multivariate kernel density
  estimators.
\newblock In {\em Annales de l'Institut Henri Poincare (B) Probability and
  Statistics\/}, vol.~38, (pp. 907--921). Elsevier.

\bibitem[{Gin{\'e} \& Nickl(2008)}]{gine2008uniform}
Gin{\'e}, E., \& Nickl, R. (2008).
\newblock Uniform central limit theorems for kernel density estimators.
\newblock {\em Probability Theory and Related Fields\/}, {\em 141\/}(3-4),
  333--387.

\bibitem[{Hall(2013)}]{hall2013bootstrap}
Hall, P. (2013).
\newblock {\em The bootstrap and Edgeworth expansion\/}.
\newblock Springer Science \& Business Media.

\bibitem[{Hogan \& Laird(1997)}]{HoganLaird97}
Hogan, J.~W., \& Laird, N.~M. (1997).
\newblock Mixture models for the joint distribution of repeated measures and
  event times.
\newblock {\em Statistics in Medicine\/}, {\em 16\/}(3), 239--257.

\bibitem[{Hoonhout \& Ridder(2018)}]{HoonhoutRidder18}
Hoonhout, P., \& Ridder, G. (2018).
\newblock {Nonignorable Attrition in Multi-Period Panels With Refreshment
  Samples}.
\newblock {\em J. Bus. Econ. Statist.\/}, {\em Forthcoming\/}.

\bibitem[{Kenward et~al.(2003)Kenward, Molenberghs, \& Thijs}]{Kenwardetal03}
Kenward, M.~G., Molenberghs, G., \& Thijs, H. (2003).
\newblock Pattern-mixture models with proper time dependence.
\newblock {\em Biometrika\/}, {\em 90\/}(1), 53--71.

\bibitem[{Kokonendji et~al.(2009)Kokonendji, Kiessé, \&
  Balakrishnan}]{Kokonendji09}
Kokonendji, C., Kiessé, T.~S., \& Balakrishnan, N. (2009).
\newblock {Semiparametric estimation for count data through weighted
  distributions}.
\newblock {\em Journal of Statistical Planning and Inference\/}, {\em
  139\/}(10), 3625--3638.

\bibitem[{Li \& Racine(2003)}]{LiRacine03}
Li, Q., \& Racine, J. (2003).
\newblock {Nonparametric estimation of distributions with categorical and
  continuous data}.
\newblock {\em Journal of Multivariate Analysis\/}, {\em 86\/}(2), 266--292.

\bibitem[{Linero(2017)}]{Linero17}
Linero, A.~R. (2017).
\newblock {Bayesian nonparametric analysis of longitudinal studies in the
  presence of informative missingness}.
\newblock {\em Biometrika\/}, {\em 104\/}(2), 327--341.

\bibitem[{Linero \& Daniels(2015)}]{LineroDaniels15}
Linero, A.~R., \& Daniels, M.~J. (2015).
\newblock {A Flexible Bayesian Approach to Monotone Missing Data in
  Longitudinal Studies With Nonignorable Missingness With Application to an
  Acute Schizophrenia Clinical Trial}.
\newblock {\em Journal of the American Statistical Association\/}, {\em
  110\/}(509), 45--55.

\bibitem[{Linero \& Daniels(2018)}]{LineroDaniels18}
Linero, A.~R., \& Daniels, M.~J. (2018).
\newblock {Bayesian Approaches for Missing Not at Random Outcome Data: The Role
  of Identifying Restrictions}.
\newblock {\em Statistical Science\/}, {\em 33\/}(2), 198--213.

\bibitem[{Little(1993)}]{Little93}
Little, R. J.~A. (1993).
\newblock Pattern-mixture models for multivariate incomplete data.
\newblock {\em J. Am. Statist. Assoc.\/}, {\em 88\/}(421), 125--134.

\bibitem[{Little(1994)}]{Little94}
Little, R. J.~A. (1994).
\newblock {A Class of Pattern-Mixture Models for Normal Incomplete Data}.
\newblock {\em Biometrika\/}, {\em 81\/}(3), 471--483.

\bibitem[{Little(1995)}]{Little95}
Little, R. J.~A. (1995).
\newblock {Modeling the Drop-Out Mechanism in Repeated-Measures Studies}.
\newblock {\em Journal of the American Statistical Association\/}, {\em
  90\/}(431), 1112--1121.

\bibitem[{Little \& Rubin(2002)}]{LittleRubin02}
Little, R. J.~A., \& Rubin, D.~B. (2002).
\newblock {\em {Statistical Analysis with Missing Data}\/}.
\newblock Hoboken, New Jersey: Wiley, 2nd ed.

\bibitem[{Liu \& Yang(2008)}]{liu2008kernel}
Liu, R., \& Yang, L. (2008).
\newblock Kernel estimation of multivariate cumulative distribution function.
\newblock {\em Journal of Nonparametric Statistics\/}, {\em 20\/}(8), 661--677.

\bibitem[{Marder \& Meibach(1994)}]{Risperidone94}
Marder, S.~R., \& Meibach, R.~C. (1994).
\newblock Risperidone in the treatment of schizophrenia.
\newblock {\em American Journal of Psychiatry\/}, {\em 151\/}(6), 825--835.

\bibitem[{Molenberghs et~al.(1998)Molenberghs, Michiels, Kenward, \&
  Diggle}]{ACMV}
Molenberghs, G., Michiels, B., Kenward, M.~G., \& Diggle, P.~J. (1998).
\newblock {Monotone missing data and pattern-mixture models}.
\newblock {\em Statistica Neerlandica\/}, {\em 52\/}(2), 153--161.

\bibitem[{Parzen(1962)}]{Parzen62}
Parzen, E. (1962).
\newblock {On Estimation of a Probability Density Function and Mode}.
\newblock {\em The Annals of Mathematical Statistics\/}, {\em 33\/}(3),
  1065--1076.

\bibitem[{Reiss(1981)}]{reiss1981nonparametric}
Reiss, R.-D. (1981).
\newblock Nonparametric estimation of smooth distribution functions.
\newblock {\em Scandinavian Journal of Statistics\/}, (pp. 116--119).

\bibitem[{Rinaldo \& Wasserman(2010)}]{Rinaldo2010}
Rinaldo, A., \& Wasserman, L. (2010).
\newblock {Generalized density clustering}.
\newblock {\em The Annals of Statistics\/}, {\em 38\/}(5), 2678--2722.
\newline\urlprefix\url{http://arxiv.org/abs/0907.3454}

\bibitem[{Robins(1997)}]{Robins97}
Robins, J.~M. (1997).
\newblock Non-response models for the analysis of non-monotone non-ignorable
  missing data.
\newblock {\em Statist. Med.\/}, {\em 16\/}(1), 21--37.

\bibitem[{Rosenblatt(1956)}]{Rosenblatt56}
Rosenblatt, M. (1956).
\newblock {Remarks on Some Nonparametric Estimates of a Density Function}.
\newblock {\em The Annals of Mathematical Statistics\/}, {\em 27\/}(3),
  832--837.

\bibitem[{Roy \& Daniels(2008)}]{RoyDaniels08}
Roy, J., \& Daniels, M.~J. (2008).
\newblock {A General Class of Pattern Mixture Models for Nonignorable Dropout
  with Many Possible Dropout Times}.
\newblock {\em Biometrics\/}, {\em 64\/}, 538--545.

\bibitem[{Rubin(1976)}]{Rubin76}
Rubin, D.~B. (1976).
\newblock Inference and missing data.
\newblock {\em Biometrika\/}, {\em 63\/}(3), 581--592.

\bibitem[{Rubin(1987)}]{Rubin87}
Rubin, D.~B. (1987).
\newblock {\em {Multiple Imputation for Nonresponse in Surveys}\/}.
\newblock Hoboken, New Jersey: Wiley.

\bibitem[{Sadinle \& Reiter(2017)}]{sadinle2017itemwise}
Sadinle, M., \& Reiter, J.~P. (2017).
\newblock Itemwise conditionally independent nonresponse modelling for
  incomplete multivariate data.
\newblock {\em Biometrika\/}, {\em 104\/}(1), 207--220.

\bibitem[{Sadinle \& Reiter(2019)}]{sadinle2019SAN}
Sadinle, M., \& Reiter, J.~P. (2019).
\newblock Sequentially additive nonignorable missing data modeling using
  auxiliary marginal information.
\newblock {\em Biometrika\/}, {\em Forthcoming\/}.

\bibitem[{Scharfstein et~al.(2018)Scharfstein, McDermott, D\'{i}az, Carone,
  Lunardon, \& Turkoz}]{Scharfstein18}
Scharfstein, D., McDermott, A., D\'{i}az, I., Carone, M., Lunardon, N., \&
  Turkoz, I. (2018).
\newblock Global sensitivity analysis for repeated measures studies with
  informative drop-out: A semi-parametric approach.
\newblock {\em Biometrics\/}, {\em 74\/}(1), 207--219.

\bibitem[{Scott(2015)}]{scott2015multivariate}
Scott, D.~W. (2015).
\newblock {\em Multivariate density estimation: theory, practice, and
  visualization\/}.
\newblock John Wiley \& Sons.

\bibitem[{Silverman(1986)}]{silverman1986density}
Silverman, B.~W. (1986).
\newblock {\em Density estimation for statistics and data analysis\/}, vol.~26.
\newblock CRC press.

\bibitem[{Thijs et~al.(2002)Thijs, Molenberghs, Michiels, Verbeke, \&
  Curran}]{Thijs}
Thijs, H., Molenberghs, G., Michiels, B., Verbeke, G., \& Curran (2002).
\newblock {Strategies to fit pattern-mixture models}.
\newblock {\em Biostatistics\/}, {\em 3\/}(2), 245--265.

\bibitem[{Titterington(1980)}]{Titterington80}
Titterington, D.~M. (1980).
\newblock {A Comparative Study of Kernel-Based Density Estimates for
  Categorical Data}.
\newblock {\em Technometrics\/}, {\em 22\/}(2), 259--268.

\bibitem[{van~der Vaart(1994)}]{van1994weak}
van~der Vaart, A. (1994).
\newblock Weak convergence of smoothed empirical processes.
\newblock {\em Scandinavian journal of statistics\/}, (pp. 501--504).

\bibitem[{van~der Vaart(1998)}]{van1998asymptotic}
van~der Vaart, A.~W. (1998).
\newblock {\em Asymptotic statistics\/}, vol.~3.
\newblock Cambridge university press.

\bibitem[{van~der Vaart \& Wellner(1996)}]{van1996weak}
van~der Vaart, A.~W., \& Wellner, J.~A. (1996).
\newblock {\em Weak convergence\/}.
\newblock Springer.

\bibitem[{Vansteelandt et~al.(2006)Vansteelandt, Goetghebeur, Kenward, \&
  Molenberghs}]{Vansteelandtetal06}
Vansteelandt, S., Goetghebeur, E., Kenward, M.~G., \& Molenberghs, G. (2006).
\newblock Ignorance and uncertainty regions as inferential tools in a
  sensitivity analysis.
\newblock {\em Statist. Sinica\/}, {\em 16\/}(3), 953--979.

\bibitem[{Wang \& Daniels(2011)}]{WangDaniels11}
Wang, C., \& Daniels, M.~J. (2011).
\newblock {A Note on MAR, Identifying Restrictions, Model Comparison, and
  Sensitivity Analysis in Pattern Mixture Models with and without Covariates
  for Incomplete Data}.
\newblock {\em Biometrics\/}, {\em 67\/}, 810--818.

\bibitem[{Wang \& Chen(2009)}]{WangChen09}
Wang, D., \& Chen, S.~X. (2009).
\newblock {Empirical likelihood for estimating equations with missing values}.
\newblock {\em Annals of Statistics\/}, {\em 37\/}(1), 490--517.

\bibitem[{Wang \& {van Ryzin}(1981)}]{WangRyzin81}
Wang, M.-C., \& {van Ryzin}, J. (1981).
\newblock {A Class of Smooth Estimators for Discrete Distributions}.
\newblock {\em Biometrika\/}, {\em 68\/}(1), 301--309.

\bibitem[{Wasserman(2006)}]{wasserman2006all}
Wasserman, L. (2006).
\newblock {\em All of nonparametric statistics\/}.
\newblock Springer Science \& Business Media.

\end{thebibliography}

\appendix

\section{Examples: CC restriction}	\label{sec::simple}



We provide a detailed analysis under the CC restriction.
Consider the case with $d$ variables. 
Let $1_d$ be a vector of ones of length $d$. The 
CC restriction assumes that the extrapolation density has the following form:
$$
f_{\cc}(x_{\bar r}\mid x_r,r) = g(x_{\bar r}\mid x_r,{1_d}) = \frac{g(x, {1_d})}{g(x_r , {1_d})}
= \frac{g(x, {1_d})}{\int g(x, {1_d})\mu(d x_{\bar r})}.
$$
To further simplify the notation, we assume all
variables are continuous and all kernel functions and smoothing bandwidths are the same. 
We use the notation
$$
K_h\left(x_r;X_{i,r}\right)\equiv \prod_{j\in r} 
K_h\left(x_j;X_{ij}\right) \equiv \prod_{j=1}^{d} K_j(x_j;X_{ij},h)^{r_{j}},
$$
where $j\in r$ stands for the indices of the elements of $r$ that are $1$
and $|r| = \sum_{j=1}^d r_j$ is the number of $1$'s in $r$.

{\bf CDF estimation.}
The estimator of the extrapolation density is
\begin{equation}
\hat{f}_{\cc,h}(x_{\bar r}\mid x_r,r) = \frac{\hat{g}_h(x, {1_d})}{\int \hat{g}_h(x, {1_d})\mu(d x_{\bar r})},
\label{eq::ccmv0}
\end{equation}
where 
\begin{align*}
\hat{g}_h(x, {1_d}) &= \frac{1}{n}\sum_{i=1}^{n}
K_h\left(x;X_{i}\right) I(R_i = {1_d}).
\end{align*}
Namely, $\hat{g}_h(x, {1_d})$ is the estimated PDF  of all study variables
using the observations without any missing entries. 
Let $\hat{F}_{\cc,h}(x_{\bar r}\mid x_r, r) = \int_{-\infty}^{x_{\bar r}}\hat{f}_{\cc,h}(x'_{\bar r}\mid x_r,r) \mu(dx'_{\bar r})$
be the estimator of $F_{\cc}(x_{\bar r}\mid x_r,r)$
using the KDE.
The CDF of the full-data distribution is estimated by 
\begin{equation}
\begin{aligned}
\hat{F}_{\cc,h}(x,r)&=
\int_{-\infty}^{x_r}\hat{F}_{\cc,h}(x_{\bar r}\mid x'_r, r)\hat{G}(dx_r',r)\\
&=\frac{1}{n}\sum_{i=1}^n \hat{F}_{\cc,h}(x_{\bar r}\mid X_{i,r}, r)
I(R_i = r, X_{i,r}\leq x_r) .
\end{aligned}
\label{eq::ccmv1}
\end{equation}



{\bf PDF estimation.}
If the goal is to estimate the PDF of the full-data distribution,
we use
\begin{equation}
\begin{aligned}
\hat{f}_{\cc,h}(x,r) &=  \hat{g}_h(x_r, r) \hat{f}_{\cc,h}(x_{\bar r}\mid x_r,r)\\
& = \frac{ \hat{g}_h(x_r, r)}{\hat{g}_h(x_r, {1_d})} \hat{g}_h(x, {1_d}) ,
\end{aligned}
\label{eq::ccmv2}
\end{equation}
where 
\begin{equation}
\begin{aligned}
\hat{g}_h(x_r, r)& = \frac{1}{n}\sum_{i=1}^{n} 
K_h\left(x_r;X_{i,r}\right) 
I(R_i = {r}).
\end{aligned}
\label{eq::ccmv3}
\end{equation}

{\bf Estimating the mean parameter.}
As a special case, we analyze the estimator for the mean parameter. 
To simplify the problem, we consider the case with two study variables $X_1$ and $X_2$
so the possible missing patterns are $\mathcal{R} = \{(00), (10), (01), (11)\}$. 
The goal is to estimate the marginal mean of $X_2$, i.e., $\mu_2= \E(X_2)$,
under the CC restriction. 
Using equation \eqref{eq::ccmv1}, 
\begin{align*}
\hat{\mu}_2 & = \int x_2 \hat{F}_{\cc,h}(dx, dr)\\
&= \sum_{r} \iint x_2 \hat{F}_{\cc,h}(dx_{\bar r}\mid x_r,r) \hat{G}(dx_r,r)\\
& = \sum_{r}\iint x_2\hat{f}_{\cc,h}(x_{\bar r}\mid x_{r},r) \mu(dx_{\bar r} )\hat{G}(dx_r, r)\\
& = \sum_{r}\iint x_2\hat{f}_{\cc,h}(x_{\bar r}\mid x_{r},r) \mu(dx_{\bar r}) \hat{G}(dx_r\mid r)\hat{g}(r)\\
& = \sum_{r}\hat{g}(r)\hat\mu_{2,r},
\end{align*}
where 
$\hat\mu_{2,r} = \iint x_2\,\,\hat{f}_{\cc,h}(x_{\bar r}\mid x_{r},r) \mu(dx_{\bar r}) \hat{G}(dx_r\mid r)$
and $\hat{g}(r) = n_r/n$, with $n_r=\sum_{i=1}^n I(R_i=r)$.


When $r=(1,1)$, we have 
$$
\hat \mu_{2,11} = \iint x_2\,\, \hat{G}(dx_1dx_2\mid 11)=\int x_2\,\, \hat{G}(dx_2\mid 11)=n_{11}^{-1}\sum_{i=1}^nX_{i,2}I(R_{i}=11).$$  
When $r=(0,0)$, we have 
\begin{align*}
\hat\mu_{2,00} &= \iint x_2\,\,\hat{f}_{\cc,h}(x_1,x_2\mid 00) \mu(dx_1dx_2)\\
&= \iint x_2\,\,\hat{g}_{h}(x_1,x_2\mid 11) \mu(dx_1dx_2)\\
&=n_{11}^{-1}\sum_{i=1}^{n} \int x_2 K_h\left(x_2;X_{i,2}\right) \mu(dx_2) I(R_i = {11})\\
&=n_{11}^{-1}\sum_{i=1}^nX_{i,2}I(R_{i}=11).
\end{align*}
When $r=(0,1)$, we have 
\begin{align*}
\hat\mu_{2,01} &= \iint x_2\,\,\hat{f}_{\cc,h}(x_{1}\mid x_{2},01) \mu(dx_1) \hat{G}(dx_2\mid 01)\\
&=\int x_2\,\, \hat{G}(dx_2\mid 01)\\
&=n_{01}^{-1}\sum_{i=1}^nX_{i,2}I(R_{i}=01).
\end{align*}
The case for $r=(1,0)$ is more involved:
\begin{align*}
\hat{\mu}_{2,10} & = \iint x_2\,\,\hat{f}_{\cc,h}(x_{2}\mid x_{1},10) \mu(dx_{2})\,\, \hat{G}(dx_1\mid 10)\\
& = \frac{1}{n_{10}}\sum_{i=1}^n\int x_2\,\,\hat{f}_{\cc,h}(x_{2}\mid X_{i,1},10) \mu(dx_{2})I(R_i=10).
\end{align*}
Recall that the CC restriction implies 
$\hat{f}_{\cc,h}(x_{2}\mid x_{1},10) = \frac{\hat{g}_{h}(x_1,x_2,11)}{\hat{g}_{h}(x_1,11)}$.
Thus,
\begin{align*}
\int x_2\,\,\hat{f}_{\cc,h}(x_{2}\mid x_{1},10) \mu(dx_{2}) &= \frac{\int x_2\hat{g}_{h}(x_1,x_2,11)\mu(dx_2)}{\hat{g}_{h}(x_1,11)}\\
&= \frac{\sum_{j=1}^n X_{j,2}K_h\left(x_1;X_{j,1}\right)I(R_j = 11)}{\sum_{k=1}^n K_h\left(x_1;X_{k,1}\right)I(R_k = 11)}\\
& = \sum_{j=1}^n X_{j,2} W_{j}(x_1),
\end{align*}
where $W_{j}(x_1) = \frac{K_h(x_1;X_{j,1})I(R_j = 11)}{\sum_{k=1}^n K_h(x_1;X_{k,1})I(R_k = 11)}$
are such that $\sum_{i=1}^nW_{i}(x_1)=1$.
With this, we can then obtain the estimator 
\begin{align*}
\hat{\mu}_{2,10} 
& = \frac{1}{n_{10}}\sum_{i,j=1}^n X_{j,2} W_{j}(X_{i,1}) I (R_i = 10)\\
& =   \frac{1}{n_{10}}\sum_{i=1}^n\sum_{j=1}^nX_{j,2}\frac{ K_h\left(X_{i,1};X_{j,1}\right)I(R_j = 11)}{\sum_{k=1}^n K_h\left(X_{i,1};X_{k,1}\right)I(R_k = 11)}I(R_i = 10)\\
& =   \frac{1}{n_{10}}\sum_{j=1}^nX_{j,2} \sum_{i=1}^n\frac{ K_h\left(X_{i,1};X_{j,1}\right)I(R_i = 10)}{\sum_{k=1}^n K_h\left(X_{i,1};X_{k,1}\right)I(R_k = 11)}I(R_j = 11).
\end{align*}

We note that $\hat{\mu}_{2,10}$ is essentially a weighted average of the fully observed data, so we may rewrite it as
$$
\hat{\mu}_{2,10} = \sum_{j=1}^n \alpha_j X_{j,2} I(R_j=11 ),
$$
where $\alpha_j = \frac{1}{n_{10}}\sum_{i=1}^n\frac{ K_h\left(X_{i,1};X_{j,1}\right)I(R_i = 10)}{\sum_{k=1}^n K_h(X_{i,1};X_{k,1})I(R_k = 11)}$. Therefore, combining all estimators $\hat{\mu}_{2,r}$ together, we conclude that 
\begin{equation}
\hat{\mu}_2 = \sum_r \hat{g}(r) \hat{\mu}_{2,r} = \sum_{i=1}^n X_{i,2} \omega_i,
\label{eq::ccmv_mu2}
\end{equation}
where 
$$
\omega_i = \begin{cases}
\frac{1+n_{00}/n_{11}+ n_{10}\alpha_i}{n},&\mbox{ if $R_i = 11$},\\
\frac{1}{n},&\mbox{ if $R_i = 01$},\\
0,&\mbox{otherwise}.
\end{cases}
$$
Note that $\omega_i\geq 0$ and $\sum_{i=1}\omega_i =1$
so the estimator is essentially a weighted estimator, where the weight is determined by
the identifying assumption.

\section{Hadamard Differentiation}	\label{sec::Hada}

The Hadamard differentiation
is one type of functional derivative
that is 
key 
for a statistical functional to have asymptotic normality. 
Here we briefly describe the definition of Hadamard differentiation.
%
%

Let $(\D_1,\|\cdot\|_{\D_1})$ and $(\D_2,\|\cdot\|_{\D_2})$ be two normed spaces
and let $\Psi:\D_1\mapsto \D_2$ be a mapping. 
$\Psi$ is called Hadamard differentiable at $\omega\in \D_1$ with a differentiation
$\dot{\Psi}_\omega: \D_1\mapsto \D_2$ if 
for any $\eta_t\rightarrow \eta$, 
$$
\lim_{t\rightarrow 0}\left\|\frac{\Psi(\omega+t\cdot \eta_t) - \Psi(\omega)}{t} -\dot{\Psi}_\omega(\eta)\right\|_{\D_2}  = 0.
$$
The Hadamard differentiation is commonly used to derive the bootstrap theory. 
Roughly speaking, 
if a statistical functional $\theta:\mcF\mapsto\R$ is Hadamard differentiable at $F_0\in \mcF$
and $\sqrt{n}(\hat{F}_n-F_0)\overset{D}{\rightarrow} \mcB$ for some random process $\mcB$, 
then 
\begin{equation}
\sqrt{n}\left(\theta(\hat{F}_n) - \theta(F_0)\right) \overset{D}{\rightarrow} \dot{\theta}_{F_0}(\mcB).
\label{eq::fdelta}
\end{equation}
Equation \eqref{eq::fdelta} is known as the functional delta method \citep{van1998asymptotic}. 
Many common population quantities can be expressed as a Hadamard-differentiable 
statistical functional; for instance, population mean, population quantiles,
population median, correlation between two variables, regression coefficients.

%
%
%
%
%

\section{Proofs}

For any function $\eta(x)$ we write 
$\eta(x)dx \equiv \eta(x)\mu(dx)$ to simplify
the notation for integration.

\begin{proof}[Proof of Theorem~\ref{thm::pdf}]

Because
$$
f_A(x_{>t}\mid x_{\leq t}, t) = \prod_{s=t+1}^d f_A(x_s\mid x_{<s}, t) = 
\prod_{s=t+1}^d g(x_s\mid x_{<s}, T\in \assum_{ts}).
$$
we have
$$
\Delta f_A(x_{>t}\mid x_{\leq t}, t)
 = \prod_{s=t+1}^d \tilde{g}(x_s\mid x_{<s}, T\in \assum_{ts})-\prod_{s=t+1}^d g(x_s\mid x_{<s}, T\in \assum_{ts}).
$$
Using the fact that 
$$
\tilde{g}(x_s\mid x_{<s}, T\in \assum_{ts}) = g(x_s\mid x_{<s}, T\in \assum_{ts}) + \Delta g(x_s\mid x_{<s}, T\in \assum_{ts}),
$$
where $\Delta g(x_s\mid x_{<s}, T\in \assum_{ts})$ is a small quantity (due to equation \eqref{eq::limit}), 
the above equality becomes
\begin{align*}
\Delta f_A(x_{>t}\mid x_{\leq t}, t)
 &= \prod_{s=t+1}^d \tilde{g}(x_s\mid x_{<s}, T\in \assum_{ts})-\prod_{s=t+1}^d g(x_s\mid x_{<s}, T\in \assum_{ts})\\
 & = \prod_{s=t+1}^d \left[{g}(x_s\mid x_{<s}, T\in \assum_{ts})+ \Delta g(x_s\mid x_{<s}, T\in \assum_{ts})\right]\\
&\qquad \qquad  -\prod_{s=t+1}^d g(x_s\mid x_{<s}, T\in \assum_{ts})\\
& =  \left\{\sum_{s=t+1}^d
\frac{\Delta g(x_s\mid x_{<s}, T\in \assum_{ts})}{ g(x_s\mid x_{<s}, T\in \assum_{ts})}\right\}\prod_{\tau=t+1}^d g(x_\tau\mid x_{<\tau}, t\in \assum_{t \tau})
+O(\Delta_1^2),
\end{align*}
where 
$O(\Delta_1^2)$ involves products of two small terms ($\Delta$ of some functions) so it is negligible. 
Note that the identifying restriction implies that 
$\prod_{\tau=t+1}^d g(x_\tau\mid x_{<\tau}, t\in \assum_{t \tau}) = f_A(x_{>t}\mid x_{\leq t}, t)$, 
so we can rewrite the above as
\begin{equation}
\frac{\Delta f_A(x_{>t}\mid x_{\leq t}, t)}{ f_A(x_{>t}\mid x_{\leq t}, t)}
 = \sum_{s=t+1}^d\frac{\Delta g(x_s\mid x_{<s}, T\in \assum_{ts})}{ g(x_s\mid x_{<s}, T\in \assum_{ts})} + O(\Delta_1^2).
 \label{eq::kde::1}
\end{equation}

Using the fact that 
$$
g(x_s\mid x_{<s}, T\in \assum_{ts}) = \frac{g(x_{\leq s}\mid T\in \assum_{ts})}{g(x_{<s}\mid T\in \assum_{ts})},
$$
we conclude that 
\begin{align*}
\Delta g(x_s\mid x_{<s}, T\in \assum_{ts})
& = \frac{\tilde{g}(x_{\leq s}\mid T\in \assum_{ts})}{\tilde{g}(x_{<s}\mid T\in \assum_{ts})}-\frac{g(x_{\leq s}\mid T\in \assum_{ts})}{g(x_{<s}\mid T\in \assum_{ts})}\\
& = \frac{{g}(x_{\leq s}\mid T\in \assum_{ts})+\Delta {g}(x_{\leq s}\mid T\in \assum_{ts})}{{g}(x_{<s}\mid T\in \assum_{ts})+\Delta {g}(x_{<s}\mid T\in \assum_{ts})}-\frac{g(x_{\leq s}\mid T\in \assum_{ts})}{g(x_{<s}\mid T\in \assum_{ts})}\\
& = 
\frac{\Delta g(x_{\leq s}\mid T\in \assum_{ts})}{g(x_{<s}\mid T\in \assum_{ts})}
- \frac{\Delta g(x_{< s}\mid T\in \assum_{ts})}{g(x_{<s}\mid T\in \assum_{ts})}\frac{g(x_{\leq s}\mid T\in \assum_{ts})}{g(x_{<s}\mid T\in \assum_{ts})}
+ O(\Delta_2^2),
\end{align*}
where $ O(\Delta_2^2)$ is again a quantity involving
the product of two or more small terms so it is negligible. 
Diving both sides by $g(x_s\mid x_{<s}, T\in \assum_{ts})$,
we conclude
$$
\frac{\Delta g(x_s\mid x_{<s}, T\in \assum_{ts})}{g(x_s\mid x_{<s}, T\in \assum_{ts})}
 = \frac{\Delta g(x_{\leq s}\mid T\in \assum_{ts})}{g(x_{\leq s}\mid T\in \assum_{ts})}
- \frac{\Delta g(x_{< s}\mid T\in \assum_{ts})}{g(x_{< s}\mid T\in \assum_{ts})} + O(\Delta_2^2).
$$
Putting this back to equation \eqref{eq::kde::1}, we conclude
$$
\frac{\Delta f_A(x_{>t}\mid x_{\leq t}, t)}{ f_A(x_{>t}\mid x_{\leq t}, t)}
= \sum_{s=t+1}^d \left\{ \frac{\Delta g(x_{\leq s}\mid T\in \assum_{ts})}{g(x_{\leq s}\mid T\in \assum_{ts})}
- \frac{\Delta g(x_{< s}\mid T\in \assum_{ts})}{g(x_{< s}\mid T\in \assum_{ts})} \right\}
+ O(\Delta_1^2+\Delta_2^2),
$$
which is the desired result by identifying
$\tilde{W}_n= O(\Delta_1^2+\Delta_2^2)$.
Note that since $O(\Delta_1^2+\Delta_2^2)$ involves product of two small
quantities,
it is easy to see that $\frac{O(\Delta_1^2+\Delta_2^2)}{\Delta f_A(x_{>t}\mid x_{\leq t}, t)}\rightarrow 0$
uniformly for all $x$.
\end{proof}


\begin{proof}[Proof of Theorem~\ref{thm::kde}]

{\bf Part 1: convergence of CDF estimator.}\\
Recall that 
\begin{align*}
\bar{F}_{A,h}(x_{>t}\mid x_{\leq t},t) = \int_{-\infty}^{x_{>t}}\bar{f}_{A,h}(x'_{>t}\mid x_{\leq t},t) dx'_{>t}.
\end{align*}
We define
$$
\hat{F}^\dagger_{A,h}(x,t) = \sum_{t} \frac{1}{n}\sum_{i=1}^n \bar{F}_{A,h}(x_{>t}\mid X_{i,r},t) 
I(T_i=t, X_{i,\leq t}\leq x_{\leq t})
$$
for each $t$.
We will use this quantity to decompose the uncertainty of $\hat{F}_{A,h}(x)$.
Specifically, We bound the rate of $\hat{F}_{A,h}-\bar{F}_{A,h}$ using the decomposition 
$$
\hat{F}_{A,h}-\bar{F}_{A,h} = \underbrace{\hat{F}_{A,h}-\hat{F}^\dagger_{A,h}}_{=(I)}+\underbrace{\hat{F}^\dagger_{A,h}-\bar{F}_{A,h}}_{=(II)}
$$
and control the rates of (I) and (II). 
After deriving this rate, we will then analyze the bias $\bar{F}_{A,h} - F_A$.

Because $\hat{F}^\dagger_{A,h}$ is the summation of IID random elements,
the rate of (II) is easier to derive so we focus on it first.
Because of the IID assumption,
\begin{align*}
\E\bigg\{\frac{1}{n}\sum_{i=1}^n &\bar{F}_{A,h}(x_{>t}\mid X_{i,\leq t},t) I(T_i=t, X_{i,\leq t}\leq x_{\leq t})\bigg\} \\
&= 
\E\left\{\bar{F}_{A,h}(x_{>t}\mid X_{1,\leq t},t) I(T_1=t, X_{1,\leq t}\leq x_{\leq t})\right\} \\
& =\int_{-\infty}^{x_{ \leq t}} \int_{-\infty}^{x_{>t}}\bar{f}_{A,h}(x'_{>t}\mid x'_{\leq t},t) g(x'_{\leq t}, t) dx'_{>t}dx'_{\leq t}\\
& =\int_{-\infty}^{x'_{ \leq t}=x_{ \leq t}} \int_{x_{ \leq t}'=-\infty}^{x_{>t}}\bar{f}_{A,h}(x'_{>t}\mid x'_{\leq t},t) G(dx'_{\leq t}, t) dx'_{>t}\\
&= \bar{F}_{A,h}(x, t)
\end{align*}
for each $t=1,\cdots, d$. 
Therefore,
$$
\E(\hat{F}^\dagger_{A,h}) =  \bar{F}_{A,h}
$$
so we only need to focus on the variance. 
The variance is very easy to derive because the variance of each individual
$$
{\sf Var}(\bar{F}_{A,h}(x_{>t}\mid X_{1,\leq t},t) I(T_1=t, X_{1,\leq t}\leq x_{\leq t})) \leq 1
$$
because both $F_A$ and indicator function are uniformly bounded by $1$.
So
${\sf Var}(\hat{F}^\dagger_{A,h}) = O(1/n)$,
which implies
$$
(II) = O_P\left(\sqrt{\frac{1}{n}}\right).
$$

To bound (I),  note that
\begin{equation}
\begin{aligned}
\hat{F}_{A,h}(x,t)-\hat{F}^\dagger_{A,h}(x,t) & = 
\frac{1}{n}\sum_{i=1}^n \left\{\hat{F}_{A,h}(x_{>t}\mid X_{i,\leq t},t)-\bar F_{A,h}(x_{>t}\mid X_{i,\leq t},t)\right\} \\
&\qquad\qquad \times
I(T_i=t, X_{i,\leq t}\leq x_{\leq t}).
\end{aligned}
\label{eq::pf::kde1}
\end{equation}
So the key is to bound the difference 
\begin{align*}
\hat{F}_{A,h}(x_{>t}\mid X_{i,\leq t},t)&- \bar{F}_{A,h}(x_{>t}\mid X_{i,\leq t},t)\\
& = \frac{1}{n}\sum_{i=1}^n\int_{-\infty}^{x_{>t}}\left\{\hat{f}_{A,h}(x'_{>t}\mid X_{i,\leq t} ,t)-\bar{f}_{A,h}(x'_{>t}\mid X_{i,\leq r} ,t)\right\}dx'_{>t},
\end{align*}
which boils down to the difference between the density estimator. 


Using Theorem~\ref{thm::pdf},
\begin{equation}
\begin{aligned}
\hat{f}_{A,h}(x_{>t}\mid x_{\leq t} ,t)&-\bar{f}_{A,h}(x_{>t}\mid x_{\leq t} ,t)\\
&= \bar{f}_{A,h}(x_{>t}\mid x_{\leq t}, t)\times 
\sum_{s=t+1}^d \bigg\{ \frac{ \hat{g}_h(x_{\leq s}\mid T\in \assum_{ts})- \bar{g}_h (x_{\leq s}\mid T\in \assum_{ts})}{\bar g(x_{\leq s}\mid T\in \assum_{ts})}\\
&\qquad - \frac{\hat {g}_h(x_{< s}\mid T\in \assum_{ts})- \bar g_h (x_{< s}\mid T\in \assum_{ts})}{\bar g_h(x_{< s}\mid T\in \assum_{ts})} \bigg\}
+ \tilde{W}_n(x),
\end{aligned}
\label{eq::pf::kde2}
\end{equation}
where $\tilde{W}_n(x)$ is a negligible term. 
Recall that 
\begin{align*}
\hat{g}_h(x_{\leq s}\mid T\in \assum_{ts}) & = \frac{1}{n_{ts}} \sum_{i=1}^n
K_h\left(x_{\leq s};X_{i,\leq s}\right) I(T_i \in \assum_{ts})\\
\hat{g}_h(x_{< s}\mid T\in \assum_{ts}) & = \frac{1}{n_{ts}} \sum_{i=1}^n
K_h\left(x_{< s};X_{i,< s}\right) I(T_i \in \assum_{ts})
\end{align*}
are both KDEs and $\bar {g}_h(x_{\leq s}\mid T\in \assum_{ts}) = \E(\hat{g}_h(x_{\leq s}\mid T\in \assum_{ts}))$
and $\bar {g}_h(x_{< s}\mid T\in \assum_{ts}) = \E(\hat{g}_h(x_{< s}\mid T\in \assum_{ts}))$
and $n_{ts} = \sum_{i=1}^n I(T_i \in \assum_{ts})$.

Let 
\begin{equation}
\begin{aligned}
\mathcal{E}_{1,s}(x,t) &= \bar{f}_{A,h}(x_{>t}\mid x_{\leq t}, t)\frac{ \hat{g}_h(x_{\leq s}\mid T\in \assum_{ts})- \bar{g}_h (x_{\leq s}\mid T\in \assum_{ts})}{\bar g_h(x_{\leq s}\mid T\in \assum_{ts})}\\
\mathcal{E}_{2,s}(x,t) &= \bar{f}_{A,h}(x_{>t}\mid x_{\leq t}, t)\frac{ \hat{g}_h(x_{< s}\mid T\in \assum_{ts})- \bar{g}_h (x_{< s}\mid T\in \assum_{ts})}{\bar g_h(x_{< s}\mid T\in \assum_{ts})}.
\end{aligned}
\label{eq::pf::E12}
\end{equation}
Equation \eqref{eq::pf::kde2} can be written as
$$
\hat{f}_{A,h}(x_{>t}\mid x_{\leq t} ,t)-\bar{f}_{A,h}(x_{>t}\mid x_{\leq t} ,t) =\sum_{s=t+1}^d 
\left\{\mathcal{E}_{1,s}(x,t) -\mathcal{E}_{2,s}(x,t) \right\} + \tilde{W}_n(x,t).
$$
Because the difference $\hat{F}_{A,h}(x_{>t}\mid X_{i,\leq t},t)- \bar{F}_{A,h}(x_{>t}\mid X_{i,\leq t},t)$
comes from integrating $\hat{f}_{A,h}(x_{>t}\mid x_{\leq t} ,t)-\bar{f}_{A,h}(x_{>t}\mid x_{\leq t} ,t)$,
ignoring the negligible terms
leads to 
\begin{align*}
\hat{F}_{A,h}(x_{>t}\mid X_{i,\leq t},t)&- \bar{F}_{A,h}(x_{>t}\mid X_{i,\leq t},t) =\\
&= \sum_{s=t+1}^d 
\int^{x_{>t}}_{-\infty} \mathcal{E}_{1,s}(x'_{>t}, X_{i,\leq t},t)dx'_{>t} -\int^{x_{>t}}_{-\infty} \mathcal{E}_{2,s}(x'_{>t}, X_{i,\leq t},t) dx'_{>t}.
\end{align*}
So we only need to bound each quantity $\int^{x_{>t}}_{-\infty} \mathcal{E}_{j,s}(x',t)dx'_{>t}$ for $j=1,2$
and $s\in \{t+1,\cdots, d\}$.

Because $\mathcal{E}_{j,s}(x',t)$
is essentially a KDE minus its expectation and rescale by a function,
the quantity $\int^{x_{>t}}_{-\infty} \mathcal{E}_{j,s}(x',t)dx'_{>t}$
is just the corresponding convergence rate
of using this as a CDF estimator. 
%
Using the convergence rate of the CDF estimator from the KDE \citep{reiss1981nonparametric,liu2008kernel} 
along with our assumption on the smoothing bandwidth 
$\frac{\log n}{nh^d}\rightarrow 0$,
we conclude that
\begin{align*}
\E\left(\int^{x_{>t}}_{-\infty} \mathcal{E}_{j,s}(x',t)dx'_{>t}\right) &= o\left(\sqrt{\frac{1}{n}}\right),\\
{\sf Var}\left(\int^{x_{>t}}_{-\infty} \mathcal{E}_{j,s}(x',t)dx'_{>t}\right) & = O\left({\frac{1}{n}}\right),
\end{align*}
for $s\in \{t+1,\cdots, d\}$ and $j = 1,2$ and uniformly for all $x_{\leq t}$.
Therefore,
\begin{align*}
\E\left(\hat{F}_{A,h}(x_{>t}\mid x_{\leq t},t)\right)- \bar{F}_{A,h}(x_{>t}\mid x_{\leq t},t) &= o\left(\sqrt{\frac{1}{n}}\right),\\
{\sf Var}\left(\hat{F}_{A,h}(x_{>t}\mid x_{\leq t},t)\right) &= O\left({\frac{1}{n}}\right).
\end{align*}
Note that the construction of $\mathcal{E}_{j,s}$ does not involve data points
with $T_i = t$, so the convergence rate will not change if
we are conditioning on $\{(X_i,T_i): T_i = t\}$.
Thus, using the law of total expectation and variance and applying
the above results into equation \eqref{eq::pf::kde1},
we conclude that 
$$
(I)=  \hat{F}_{A,h}(x,t )-\hat{F}^\dagger_{A,h}(x,t)= O_P\left(\sqrt{\frac{1}{n}}\right).
$$

To analyze the bias $\bar{F}_{A,h}-F_A$, note that 
the difference between $\bar{F}_{A,h}(x) - F_A(x)$
is due to the quantity 
$$
\bar{F}_{A,h}(x_{>t}\mid x_{\leq t},t) - F_A(x_{>t}\mid x_{\leq t},t) 
 = \int_{-\infty}^{x_{>t}}\left\{\bar{f}_{A,h}(x_{>t}\mid x_{\leq t},t)- f_A(x_{>t}\mid x_{\leq t},t)\right\}dx_{>t}.
$$
Theorem~\ref{thm::pdf} shows that the difference $\bar{f}_{A,h}(x_{>t}\mid x_{\leq t},t)- f_{A}(x_{>t}\mid x_{\leq t},t)$
is due to the difference $\bar{f}_h(x_{\leq s}\mid T\in \assum_{ts}) - f(x_{s}\mid T\in \assum_{ts})$,
which is the bias rate of the KDE $O(h^2)$ under the regular smoothness condition (assumption (A1) and (K1)).
Therefore, we conclude that 
$$
\bar{F}_{A,h}(x,t) - F_A(x,t) = O(h^2).
$$

{\bf Part 2: convergence of PDF estimator.}\\
Recall that 
\begin{align*}
\hat{f}_{A,h}(x,t) &=   \hat{g}(t)\hat {g}_h(x_{\leq t}\mid t)\hat{f}_{A,h}(x_{>t}\mid x_{\leq t},t)\\
\bar{f}_{A,g}(x,t) &=   {g}(t)\bar {g}_h(x_{\leq t}\mid t)\bar{f}_{A,h}(x_{>t}\mid x_{\leq t},t)\\
{f}_A(x,t) &=   {g}(t) {g}(x_{\leq t}\mid t){f}_A(x_{>t}\mid x_{\leq t},t).
\end{align*}
We first derive the bound on $\hat{f}_{A,h}(x,t)-\bar{f}_{A,h}(x,t)$.

It is clear that the difference comes from
$\hat{g}(t) - g(t)$ and $\hat {g}_h(x_{\leq t}\mid t)-\bar {g}_h(x_{\leq t}\mid t)$
and $\hat{f}_{A,h}(x_{>t}\mid x_{\leq t},t)-\bar{f}_{A,h}(x_{>t}\mid x_{\leq t},t)$. 
The first quantity is just the population proportion versus the empirical ratio
so the rate is $\hat{g}(t) - g(t) = O_P\left(\sqrt{1/n}\right)$.
The second quantity is the classical result about the KDE versus its expectation
so the convergence rate is 
$$
\hat {g}_h(x_{\leq t}\mid t)-\bar {g}_h(x_{\leq t}\mid t) = O_P\left(\sqrt{\frac{1}{nh^{t}}}\right),
$$
see, e.g. \cite{wasserman2006all} and \cite{scott2015multivariate}. 
The last quantity can be reduced to the problem of KDE versus its expectation
using Theorem~\ref{thm::pdf}
and thus, we conclude
\begin{align*}
\hat{f}_{A,h}(x_{>t}\mid x_{\leq t},t)-\bar{f}_{A,h}(x_{>t}\mid x_{\leq t},t) 
& = \sum_{s=t+1}^dO_P(\hat{g}_h(x_{\leq s}\mid T\in \assum_{ts})- \bar{g}_h (x_{\leq s}\mid T\in \assum_{ts}))\\
&\qquad+O_P(\hat{g}_h(x_{< s}\mid T\in \assum_{ts})- \bar{g}_h (x_{< s}\mid T\in \assum_{ts}))\\
&= \sum_{s=t+1 }^dO_P\left(\sqrt{\frac{1}{nh^{s}}}\right)\\
& = O_P\left(\sqrt{\frac{1}{nh^{d}}}\right).
\end{align*}
Therefore, after combing all the difference, the dominant term is $O_P\left(\sqrt{\frac{1}{nh^{d}}}\right)$
so we conclude
$$
\hat{f}_{A,h}(x,t)-\bar{f}_{A,h}(x,t) = O_P\left(\sqrt{\frac{1}{nh^{d}}}\right). 
$$

The analysis on $\bar{f}_{A,h}(x,t)- f_A(x,t)$ is just the bias analysis of the KDE.
Using the bias of the KDE and the smoothness condition (A1),
$\bar {g}_h(x_{\leq t}\mid t)- {g}(x_{\leq t}\mid t) = O(h^2)$.
Moreover, by applying Theorem~\ref{thm::pdf}, 
the the difference $\bar{f}_{A,h}(x_{>t}\mid x_{\leq t},t)-{f}_{A,h}(x_{>t}\mid x_{\leq t},t)$
is the difference between KDEs and their expectations,
which is the bias of the KDE.
So we conclude $\bar{f}_{A,h}(x_{>t}\mid x_{\leq t},t)-{f}_A(x_{>t}\mid x_{\leq t},t) = O(h^2)$
which completes the proof.
\end{proof}

\begin{proof}[Proof of Theorem~\ref{thm::uclt}]
First note that the difference can be written as
\begin{align*}
\hat{F}_{A,h}(x,t) - \bar{F}_{A,h}(x,t) &=\bigg\{
\frac{1}{n}\sum_{i=1}^n \hat{F}_{A,h}(x_{> t}\mid X_{i,\leq t},t)\times\\
&\qquad I(T_i=t, X_{i,\leq t}\leq x_{\leq t}) 
- \E\left(\bar{F}_{A,h}(x_{>t}\mid X_{1,\leq t},t)\right)
\bigg\}\\
& =  \bigg\{
\int^{x_{\leq t}' = x_{\leq t}}_{x_{\leq t}' = -\infty}\hat{F}_{A,h}(x_{> t}\mid x'_{\leq t},t)\hat{G}(dx_{\leq t}'\mid t) \\
&\qquad\qquad-\int^{x_{\leq t}' = x_{\leq t}}_{x_{\leq t}' = -\infty}\bar{F}_{A,h}(x_{> t}\mid x_{\leq t}',r) G(dx_{\leq t}', t)
\bigg\}\\
& =  \mathcal{E}'_{n,t}(x).
\end{align*}
Thus, we will focus on a given $t$ and show that $\sqrt{n}\mathcal{E}'_{n,t}$ converges to
a Brownian bridge.

To simplify the problem, we define
\begin{align*}
\Delta_{1,n}(x_{\leq t}, t) & = \hat{G}(x_{\leq t}, t) - G(x_{\leq t}, t),\\
\Delta_{1,n}(dx_{\leq t}, t) & = \hat{G}(dx_{\leq t}, t) - G(dx_{\leq t}, t),\\
\Delta_{2,n}(x_{> t}\mid x_{\leq t},t) &= \hat{F}_{A,h}(x_{> t}\mid x_{r},t)- \bar{F}_{A,h}(x_{> t}\mid x_{\leq t},t),\\
\Delta_{2,n}(dx_{> t}\mid x_{\leq t},t) &= \hat{F}_{A,h}(dx_{> t}\mid x_{r},t)- \bar{F}_{A,h}(dx_{> t}\mid x_{\leq t},t).
\end{align*}

The quantity $\Delta_{1,n}(x_{\leq t}, t)$ is just the difference between
an empirical distribution function and the corresponding CDF
so the DKW inequality \citep{dvoretzky1956asymptotic} implies 
\begin{equation}
\sup_{x_{\leq t},t} \left|\Delta_{1,n}(x_{\leq t}, t)\right| = O_P\left(\sqrt{\frac{1}{n}}\right).
\label{eq::D1}
\end{equation}
Also, because $\Delta_{2,n}(x_{> t}\mid x_{\leq t},t)$ is related to the smoothed CDF estimator versus
its expectation,
it is known that 
\begin{equation}
\sup_{x_{> t}, x_{\leq t},t} \left|\Delta_{2,n}(x_{> t}\mid x_{\leq t},t)\right| = O_P\left(\sqrt{\frac{\log n}{n}}\right);
\label{eq::D2}
\end{equation}
see, e.g., \cite{reiss1981nonparametric} and \cite{liu2008kernel}. 

Using $\Delta_{1,n}$ and $\Delta_{2,n}$, we can rewrite $\sqrt{n}\mathcal{E}'_{n,t}$ as
\begin{align*}
\sqrt{n}\mathcal{E}'_{n,t} &= \sqrt{n}\int^{x_{\leq t}' = x_{\leq t}}_{x_{\leq t}' = -\infty}
\hat{F}_{A,h}(x_{> t}\mid x'_{\leq t},t)\hat{G}(dx_{\leq t}', t) 
-\int^{x_{\leq t}' = x_{\leq t}}_{x_{\leq t}' = -\infty}\bar{F}_{A,h}(x_{> t}\mid x_{\leq t}',t) G(dx_{\leq t}', t)\\
& = \underbrace{\int^{x_{\leq t}' = x_{\leq t}}_{x_{\leq t}' = -\infty} \bar{F}_{A,h}(x_{> t}\mid x_{\leq t}',t)
\sqrt{n}\Delta_{1,n}(dx_{\leq t}', t)}_{(I)}\\
&\qquad+\underbrace{\int^{x_{\leq t}' = x_{\leq t}}_{x_{\leq t}' = -\infty}\sqrt{n} \Delta_{2,n}(x_{\leq t}\mid x_{\leq t}',t) G(dx_{\leq t}', t)}_{(II)}\\
&\qquad \qquad+ \underbrace{\sqrt{n} \int^{x_{\leq t}' = x_{\leq t}}_{x_{\leq t}' = -\infty}\Delta_{2,n}(x_{\leq t}\mid x_{\leq t}',t) \Delta_{1,n}(dx_{\leq t}', t)}_{(III)}.
\end{align*}

To show the convergence toward a Brownian bridge, we separately analyze
each term.
Due to equation \eqref{eq::D1} and \eqref{eq::D2}, $(III) = o_P(1) $
so we only need to focus on (I) and (II).

{\bf Analysis of (I):}\\
Observe that the quantity 
$$
\sqrt{n}\Delta_{1,n}(x_{\leq t}, t)  = \sqrt{n}\left(\hat{G}(x_{\leq t}, t) - G(x_{\leq t}, t) \right) 
= \mathbb{G}_t (x_{\leq t})
$$
defines an empirical process. 
Thus, (I) can be written as
$$
(I) = \int \bar{F}_{A,h}(x_{> t}\mid x_{\leq t}',t) I(x_{\leq t}'\leq x_{\leq t}) \mathbb{G}_t (dx'_{\leq t}) 
= \mathbb{G}_t(f_{x_{> t}, x_{\leq t}}),
$$
where the last equality is the common expression in the empirical process theory 
\citep{van1996weak, van1998asymptotic}
and $f_{x_{> t}, x_{\leq t}}:\mathbb{R}^{t}\mapsto \mathbb{R} $ is the function 
$f_{x_{> t}, x_{\leq t}}(y) = \bar{F}_{A,h}(x_{> t}\mid y,t) I(y\leq x_{\leq t}) $
with two indices $x_{> t}, x_{\leq t}$.
To show that $(I)$ converges to a Brownian bridge,
we need to show that the class of functions
$$
\mathcal{F}_{0,h}= \left\{f_{x_{> t}, x_{\leq t}}(y) = \bar{F}_{A,h}(x_{> t}\mid y,t) I(y\leq x_{\leq t}) : (x_{> t},x_{\leq t})\in\mathcal{X}\right\}
$$
has some nice property for every $h\leq 1$.
Let $\mathcal{X}_{> t}$ be the support of $x_{> t}$ and $\mathcal{X}_{\leq t}$ be the support of $x_{\leq t}$.
Due to assumption (A1), the derivative of $\bar{F}_{A,h}(x_{> t}\mid y,t)$
with respect to $x_{> t}$ is uniformly bounded
so the class 
$$
\mathcal{F}_{1,h}= \left\{f_{x_{> t}}(y) = \bar{F}_{A,h}(x_{> t}\mid y,t) : x_{> t}\in\mathcal{X}_{> t}\right\}
$$
has an $\epsilon$-bracketing number of the order $O(1/\epsilon^{d-t})$ (see example 19.7 of \citealt{van1998asymptotic})
so its uniform entropy integral converges (loosely speaking, $\int_0^{\delta}\log \mbox{($\epsilon$-bracketing number)} d\epsilon$ converges to $0$ when $\delta\rightarrow 0$) 
for every $h\leq 1$.
Moreover, 
$$
\mathcal{F}_2= \left\{f_{x_{\leq t}}(y) =I(y\leq x_{\leq t}) : x_{ \leq t}\in\mathcal{X}_{ \leq t}\right\}
$$
has an $\epsilon$-bracketing number of the order $O(1/\epsilon^{2t})$ 
(see example 19.7 of \citealt{van1998asymptotic}; the power is $2t$ because we have $t$ variables)
so again its uniform entropy integral converges. 
Thus, the class $\mathcal{F}_{0,h} =\{f_1\cdot f_2: f_1\in \mathcal{F}_{1,h}, f_2\in \mathcal{F}_2\}$
has a covering number shrinking at a polynomial rate of $\epsilon$
so the uniform entropy integral converges for every $h\leq 1$. 
%
Therefore, according to Theorem 19.28 in \cite{van1998asymptotic}, $(I)$ converges to a Brownian bridge.

{\bf Analysis of  (II):}
Recalled that from Theorem~\ref{thm::pdf},
\begin{align*}
\hat{f}_{A,h}&(x_{> t}\mid x_{\leq t} ,t)-\bar{f}_{A,h}(x_{> t}\mid x_{\leq t} ,t)\\
&= \bar{f}_{A,h}(x_{>t}\mid x_{\leq t}, t)\times 
\sum_{s=t+1}^d \bigg\{ \frac{ \hat{g}_h(x_{\leq s}\mid T\in \assum_{ts})- \bar{g}_h (x_{\leq s}\mid T\in \assum_{ts})}{\bar g_{h}(x_{\leq s}\mid T\in \assum_{ts})}\\
&\qquad - \frac{\hat {g}_h(x_{< s}\mid T\in \assum_{ts})- \bar g_h (x_{< s}\mid T\in \assum_{ts})}{\bar g_h(x_{< s}\mid T\in \assum_{ts})} \bigg\}
+ \tilde{W}_n(x,t)\\
&=\sum_{s=t+1}^d 
\left\{\mathcal{E}_{1,s}(x,t) -\mathcal{E}_{2,s}(x,t) \right\} + \tilde{W}_n(x,t)
\end{align*}
where $\tilde{W}_n(x,t)$ is a negligible term
and $\mathcal{E}_{1,s}(x,t), \mathcal{E}_{2,s}(x,t)$ are defined in equation \eqref{eq::pf::E12}.
Because $\tilde{W}_n(x,t)$ is negligible, we ignore it in our analysis.
After ignoring it, we obtain
\begin{align*}
\Delta_{2,n}(x_{> t}\mid x_{\leq t},t) &= \hat{F}_{A,h}(x_{> t}\mid x_{r},t)- \bar{F}_{A,h}(x_{> t}\mid x_{\leq t},t)\\
& = \int_{x'_{> t} = -\infty}^{x'_{>t}=x_{> t}} \left\{\hat{f}_{A,h}(x'_{>t }\mid x_{\leq t} ,t)-\bar{f}_{A,h}(x'_{> t}\mid x_{\leq t} ,t)\right\}
dx'_{> t}\\
& = \int_{x'_{> t} = -\infty}^{x'_{> t}=x_{> t}}\sum_{s=t+1}^d \left\{\mathcal{E}_{1,s}(x',t) -\mathcal{E}_{2,s}(x',t) \right\}
dx'_{> t}\\
& = \sum_{s=t+1}^d\int_{x'_{> t} = -\infty}^{x'_{> t}=x_{> t}} \left\{\mathcal{E}_{1,s}(x',t) -\mathcal{E}_{2,s}(x',t) \right\}
dx'_{> t}
\end{align*}

Using the fact that $G(dx'_{\leq t}, t) = g(x'_{\leq t}, t)dx'_{\leq t}$ and the above expression, 
the quantity (II) can be written as
\begin{equation}
\begin{aligned}
(II) & = \int^{x_{\leq t}' = x_{\leq t}}_{x_{\leq t}' = -\infty}\sqrt{n} \Delta_{2,n}(x_{\leq t}\mid x_{r}',t) G(dx_{\leq t}', t)\\
&=\sqrt{n}\sum_{s=t+1}^d
\int^{x_{\leq t}' = x_{\leq t}}_{x_{\leq t}' = -\infty} \int_{x'_{> t} = -\infty}^{x'_{> t}=x_{> t}} \left\{\mathcal{E}_{1,s}(x',t) -\mathcal{E}_{2,s}(x',t) \right\}
dx'_{> t}G(dx_{\leq t}', t)\\
& = \sum_{s=t+1}^d \Omega_{1,s}(x,t)-\Omega_{2,s}(x,t).
\end{aligned}
\label{eq::SEP}
\end{equation}
Since the result is an additive form of each quantity involving $j$ and $s$,
we will show that 
\begin{equation}
\Omega_{j,s}(x,t) = \sqrt{n}
\int^{x_{\leq t}' = x_{\leq t}}_{x_{\leq t}' = -\infty} \int_{x'_{> t} = -\infty}^{x'_{> t}=x_{> t}}\mathcal{E}_{j,s}(x',t)
dx'_{> t}G(dx_{\leq t}', t)
\end{equation}
converges to a Brownian bridge for each $j=1,2$ and $s\in \{t+1,\cdots, d\}$.

For simplicity, we focus on $j=1$ case. 
The case of $j=2$ can be derived similarly. 
Using the fact that $G(dx_{\leq t}', t) = g(x_{\leq t}', t)dx'_{\leq t}$,
we can rewrite $\Omega_{j,s}(x)$  as
\begin{align*}
\Omega_{1,s}(x,t) & = \sqrt{n} \int_{x'=-\infty}^{x'=x}
\mathcal{E}_{1,s}(x',t) g(x_{\leq t}', t) dx'_{\leq t}\\
&= \sqrt{n} \int_{x'=-\infty}^{x'=x}
\frac{\bar{f}_{A,h}(x'_{>t}\mid x'_{\leq t}, t)g(x_{\leq t}', t)}{g(x_{\leq s}\mid T\in \assum_{ts})}\times 
\left(\hat{g}_h(x'_{\leq s}\mid T\in \assum_{ts} )-\bar{g}_h(x'_{\leq s}\mid T\in \assum_{ts} )\right)dx'\\
& = \sqrt{n}\int_{x'=-\infty}^{x'=x} \omega_{0,h}(x', s, t) \left(\hat{g}_h(x'_{\leq s}\mid T\in \assum_{ts} )-\bar{g}_h(x'_{\leq s}\mid T\in \assum_{ts} )\right)dx'\\
& = \sqrt{n}\int_{x'_{\leq s}=-\infty}^{x'_{\leq s}=x_{\leq s}} \omega_{1,h}(x_{>s}, x'_{\leq s}, s, t) \left(\hat{g}_h(x'_{\leq s}\mid T\in \assum_{ts} )-\bar{g}_h(x'_{\leq s}\mid T\in \assum_{ts} )\right)dx'_{\leq s},
\end{align*}
where 
\begin{align*}
\omega_{0,h}(x', s, t)& = \frac{\bar{f}_{A,h}(x'_{>t}\mid x'_{\leq t}, t)g(x_{\leq t}', t)}{g(x_{\leq s}\mid T\in \assum_{ts})}\\
\omega_{1,h}(x_{>s}, x'_{\leq s}, s, t) &= \int_{x'_{>s} = -\infty}^{x'_{>s} = x_{>s}}  \omega_{0,h}(x', s, t)  dx'_{>s}.
\end{align*}

For a function $f_{z}: \mathbb{R}^s \mapsto \mathbb{R}$ with index $z\in \mathbb{T}$, 
we define
a smoothed empirical process $\tilde{\mathbb{G}}_s$ 
such that 
$$
\tilde{\mathbb{G}}_s(f_z) = \sqrt{n}\int f_z(x_{\leq s}) \left(\hat{g}_h(x_{\leq s}\mid T\in \assum_{ts})-\bar{g}_h(x_{\leq s}\mid T\in\assum_{ts})\right)dx_{\leq s}.
$$
Let $\mathcal{F} = \{f_z: z\in \mathbb{T}\}$. 
By Theorem 2 and  Section 4 of \cite{gine2008uniform}, 
the smoothed empirical process defined on $\mathcal{F}$ converges to 
a Brownian bridge if the the function space $\mathcal{F}$
is Donsker (there are assumptions on the kernel functions and smoothing bandwidth;
assumption (K1-2) satisfy the assumptions).
In our case, the function space is
\begin{align*}
\mathcal{G}_{t,s,h} &= \{g_x(y): x\in\mathcal{X}\}\\
g_x(y)&=\omega_{1,h}(x_{>s}, y, s, t) I(y\leq x_{\leq s}).
\end{align*}
Namely, 
$$
\Omega_{1,s}(x) = \tilde{\mathbb{G}}_r(g_x).
$$
So we only need to show that $\mathcal{G}_{t,s,h}$
is a Donsker class for every $h\leq 1$.

To show that $\mathcal{G}_{t,s,h}$ is Donsker,
first note that any function $g_x\in \mathcal{G}_{t,s,h}$
can be written as
$$
g_x(y) = \psi_x(y)\cdot q_x(y),
$$
where $\psi_x(y) = I(y\leq x_{\leq s})$ and $q_x(y) = \omega_{1,h}(x_{>s}, y, s, t) $
so 
$$
\mathcal{G}_{t,s,h} = \mathcal{H}\times \mathcal{Q}_{s,t,h},
$$
with $\mathcal{H} = \{\psi_x(y): x\in \mathcal{X}\}$ 
is a collection of indicator functions
and $\mathcal{Q}_{s,t,h} = \{q_x(y): x\in \mathcal{X}\}$
is a collection of smooth functions with a bounded derivative (due to assumption (A1)--
note that $\omega_{1,h}$ are product of densities with integral so it has a bounded derivative).
Thus, both $\mathcal{H}$ and $\mathcal{Q}_{s,t,h}$ for every $h\leq 1$ are Donsker 
so $\mathcal{G}_{t,s,h} = \mathcal{H}\times \mathcal{Q}_{s,t,h}$
is also Donsker for every $h\leq 1$.
Thus,
we have shown that $\Omega_{1,s}(x)$
converges to
a Brownian bridge for each $s$.
The same argument works for $\Omega_{2,s}(x)$
so 
$(II)= \sum_{s=t+1}^d \Omega_{1,s}(x)-\Omega_{2,s}(x)$
converges to a Brownian bridge.

Putting the analysis of (I) and (II) together and use the fact that (III) is negligible, 
we conclude that $\sqrt{n}(\hat{F}_{A,h}(x,t) - \bar{F}_{A,h}(x,t))$ converges
to a Brownian bridge for each $t$.
\end{proof}

%
%
%

%

\begin{proof}[Proof of Theorem~\ref{thm::MC_icin}]
Because each $X^{(v)}_{i,>T_i}$ is generated independently and identically from
each other, we only need to show that they are from the distribution function 
$\hat{F}_{A,h}(x_{>T_i}\mid X_{i,\leq T_i}, t=T_i)$. 

Recall that from Algorithm \ref{alg::MC}, each $X^{(v)}_{i,>T_i}$
is created by the ordering $X^{(v)}_{i,T_i+1}, X^{(v)}_{i,T_i+2},\cdots, X^{(v)}_{i,d}$. 
Thus, we prove this by induction. 
To start with, we consider the case $s=T_i+1$ and  
we show that $X^{(v)}_{i,T_i+1}$ is from the distribution function 
$$
\hat{F}_{A,h}(x_{T_i+1}\mid X_{i,< T_i+1}, t=T_i)
$$
or its PDF $\hat{f}_{A,h}(x_{T_i+1}\mid X_{i,\leq T_i}, t=T_i)$. 

By Algorithm \ref{alg::MC}, $X^{(v)}_{i,T_i+1}$ is obtained by
first sampling an index $\ell \in \{1,\cdots, n\}$ such that $P(\ell = N\mid {\sf data}) = W_{N}(X^{(v)}_{<T_i+1};\assum_{t,T_i+1})$
and then drawing from the PDF $K_{T_i+1}(\cdot;X_{\ell, T_i+1}, h_{T_i+1})$. 
This is essentially sampling from a mixture distribution
so the PDF of $X^{(v)}_{i,T_i+1}$ is
$\sum_{\ell=1}^n W_{\ell}(X^{(v)}_{<T_i+1};\assum_{t,T_i+1})K_{T_i+1}(\cdot;X_{\ell, T_i+1}, h_{T_i+1})$,
which by equation \eqref{eq:est_monotone_donor} and \eqref{eq:gmixture}
is $\hat{f}_{A,h}(\cdot\mid X_{i,< T_i+1}, t=T_i)$. 
Thus, we have proved the case of $T_i+1$.

Now assume that it is true for all $s = T_i+1,\cdots, \tau$. 
Namely, $X^{(v)}_{i,s}$ has a PDF $
\hat{f}_{A,h}(x_{s}\mid X_{i,<s}, t=T_i)
$
so the joint PDF is $\prod_{s=T_i+1}^{\tau}\hat{f}_{A,h}(x_{s}\mid X_{i,< s}, t=T_i)$. 
We will show that it implies that $s=\tau+1$ is also true.
By the same derivation as the case of $s=T_i+1$,
one can easily see that conditioned on $X^{(v)}_{i,T_i+1},\cdots, X^{(v)}_{i,\tau}$,
$X^{(v)}_{i,\tau+1}$ is from a mixture distribution
whose PDF is $$
\sum_{\ell=1}^n W_{\ell}(X^{(v)}_{<\tau+1};\assum_{t,\tau+1})K_{\tau+1}(\cdot;X_{\ell, \tau+1}, h_{\tau+1})
= \hat{f}_{A,h}(\cdot\mid X_{i,\leq \tau}, t=T_i).
$$
Thus, the joint PDF of $X^{(v)}_{i,T_i+1},\cdots, X^{(v)}_{i,\tau+1}$
is the product
$$
\prod_{s=T_i+1}^{\tau+1}\hat{f}_{A,h}(x_{s}\mid X_{i,< s}, t=T_i)
$$
so the result holds for $s=\tau+1$ case. 
By induction, it holds for every $s=T_i+1,\cdots, d$,
which proves that 
$X^{(v)}_{i,>T_i}$
has a PDF
$$
\prod_{s=T_i+1}^{d}\hat{f}_{A,h}(x_{s}\mid X_{i,< s}, t=T_i) = \hat{f}_{A,h}(x_{>T_i}\mid X_{i,\leq T_i}, t=T_i),
$$
which is the desired result.
\end{proof}

\begin{proof}[Proof of Theorem~\ref{thm::LBT}]
With the convergence toward a Brownian bridge (Theorem~\ref{thm::uclt}),
this result follows from the Theorem of the bootstrap for delta method; see, e.g., Theorem 23.9 of \cite{van1998asymptotic}
and Theorem 3.9.11 of \cite{van1996weak}.
\end{proof}

\end{document}